\documentclass[11pt]{article}
\usepackage[utf8]{inputenc}
\usepackage{amssymb}
\usepackage{amsmath,setspace,geometry}
\usepackage{amsthm}
\usepackage{amsfonts}
\usepackage[shortlabels]{enumitem}
\usepackage{rotating}
\usepackage{pdflscape}
\usepackage{graphicx}
\usepackage{bbm}
\usepackage{comment}
\usepackage[dvipsnames]{xcolor}
\usepackage{hyperref}
\hypersetup{colorlinks=true, linkcolor= BrickRed, citecolor = BrickRed, filecolor = BrickRed, urlcolor = BrickRed, hypertexnames = true}
\usepackage[]{natbib} 
\bibpunct[:]{(}{)}{,}{a}{}{,}
\usepackage[english]{babel}
\usepackage{float}
\usepackage{placeins}
\usepackage{caption}
\usepackage{subcaption}
\usepackage{booktabs}
\usepackage{pdfpages}
\usepackage{threeparttable}
\usepackage{lscape}
\usepackage{bm}
\usepackage{tikz}
\usetikzlibrary{positioning, shapes, arrows}
\usepackage{setspace}
\usepackage{csquotes}

\newcommand{\R}{\mathbb R}

\newcommand{\Mn}{\mathcal M^n}
\newcommand{\Mnpos}{\mathcal M^n_{>0}}
\newcommand{\Nn}{\mathcal N^n}
\newcommand{\one}{\mathbf J}
\newcommand{\eps}{\varepsilon}
\newcommand\cites[1]{\citeauthor{#1}'s\ (\citeyear{#1})}
\DeclareMathOperator{\tr}{tr}

\theoremstyle{plain}
\newtheorem{theorem}{Theorem}
\newtheorem{proposition}{Proposition}
\newtheorem{corollary}{Corollary}

\newtheorem*{theoremCDMZthree}{Theorem CDMZ3}

\theoremstyle{definition}

\newtheorem*{remark}{Remark}

\usepackage{color}

\newcommand{\1}{\textup{\mbox{1}\hspace{-0.25em}\mbox{l}}}

\begin{document}

\title{%
Measuring Assortativeness: Axiomatic Foundations and Evidence from a Japanese Marriage Agency Platform\thanks{We thank Yu Awaya, Keisuke Bando, Pierre-Andr\'e Chiappori, Kaito Fujii, Alfred Galichon, Toshiyuki Hirai, Daisuke Hirata, Yusuke Ishihata, Chihiro Inoue, Michihiro Kandori, Yasushi Kawase, Fuhito Kojima, Shunya Noda, Kenji Okubo, Kyohei Okumura, Kentaro Tomoeda, as well as participants at the 18th Meeting of the Society for Social Choice and Welfare, the European Meeting on Game Theory 2026 (SING21), and the Asian Game Theory Conference 2026, for their helpful comments. We acknowledge Shuto Fukuda for sharing the data and their technical and institutional knowledge of the IBJ platform. This paper subsumes ``A Note on Assortativeness Measures'' (UTMD Working Paper 105, December 2025). This work was supported by JST ERATO Grant Number JPMJER2301, JSPS KAKENHI Grant Numbers JP23K12443 and 26KJ0926. During the preparation of this manuscript, we used OpenAI's GPT models and Anthropic's Claude models to support the development of some mathematical proofs and to assist with writing. The authors retain full responsibility for all aspects of the paper, particularly the validity of its mathematical results and proofs.}}
\author{Kenzo Imamura\thanks{\href{mailto:}{imamurak@e.u-tokyo.ac.jp} Graduate School of Economics, The University of Tokyo.
 },
 Suguru Otani\thanks{\href{mailto:}{suguru.otani@e.u-tokyo.ac.jp} Graduate School of Economics, The University of Tokyo.},
 Tohya Sugano\thanks{\href{mailto:}{sugano-tohya1011@g.ecc.u-tokyo.ac.jp} Graduate School of Economics, The University of Tokyo.},
 Koji Yokote\thanks{\href{mailto:}{koji.yokote@gmail.com} College  of Economics, Aoyama Gakuin University.}
 }
 \date{First version: December 25, 2025 \\ This version: \today} 
\maketitle

\begin{abstract}
This paper studies the axiomatic foundations of measures of assortative matching and applies these measures to real-world data. We show that an existing axiomatization of a measure called the aggregate likelihood ratio is not valid as stated.  We then propose two modifications that recover the intended axiomatization. We further extend the axiomatization to weighted aggregate likelihood ratios. We also develop corresponding measures for many-to-many matching markets. Finally, using data from a Japanese marriage agency platform, we show how assortative matching evolves across stages of the partner-selection process.
\end{abstract}

\section{Introduction}
\label{sec:introduction}

Who matches with whom is a central question in two-sided markets such as marriage and labor markets. The degree of assortative matching---how similar partners are in education, age, income, and other characteristics---has important implications for inequality and welfare: if high-income or highly educated individuals increasingly marry one another, household income inequality may rise even when individual wage distributions remain stable \citep{greenwood2014marry,eika2019educational}. In many of these markets, final one-to-one matches emerge from earlier stages in which each agent interacts with many potential partners, through dates before marriage or applications and interviews before hiring. The economics literature on measuring assortative matching has largely focused on one-to-one markets, and data rarely follow the same market from initial contacts to final matches, so the stage of partner selection at which sorting emerges is rarely measured directly.

This paper (i) develops an axiomatic analysis of assortative matching measures for both one-to-one and many-to-many matching markets and (ii) applies these measures to data from a Japanese marriage agency platform that records every stage of the partner-selection process. Our analysis builds on the framework developed by \citet{chiappori2025changes}, who provide axiomatic characterizations of measures of assortative matching and empirical evidence on changes in marital sorting in the United States. Their framework offers a systematic approach to comparing matching patterns across populations and over time. We revisit and extend this framework through four contributions: three theoretical and one empirical.

First, we provide a counterexample showing that \cites{chiappori2025changes} axiomatization of the \emph{aggregate likelihood ratio (ALR)} is not valid as stated. 
ALR is a commonly used measure of assortativeness,\footnote{See
the first paragraph of the ``Related literature'' section for
previous studies using this measure.} defined as the ratio of
the observed number of same-type pairs to the number expected
if men and women were matched completely at random. 
Their axioms do not characterize ALR up to a positive multiplicative
constant: they also admit indices whose numerators assign positive
weight to different-type pairs. We offer two modifications. 
The first retains their original axioms and adds continuity and a condition placing matchings with no same-type pairs at the bottom of the assortativeness ranking, recovering the intended axiomatization (Theorem~\ref{thm:multitype-alr}). 
The second concerns their {\it decomposability axiom}, which governs
how the index behaves when populations are pooled. We replace this
axiom with a new one, {\it Random-Benchmark Monotonicity}: the index weakly increases when the observed number of same-type pairs weakly increases although there are fewer opportunities for such pairs to form. Together with additional axioms, it yields axiomatizations of positive affine transformations of ALR (Theorems~\ref{thm:alr-without-rd} and~\ref{thm:marginal-random-decomposability}).

Second, we provide an axiomatization of the weighted ALRs introduced by \citet{choi2025sorting}, which account for similarities between different types. Such similarities are particularly important when there are many types. For example, when age is divided into one-year groups, the ordinary ALR treats a match between ages 30 and 31 in the same way as a match between ages 30 and 50, despite the difference in age proximity. Under a natural restriction on similarity weights, a weighted version of Random-Benchmark Monotonicity, together with the other axioms, characterizes weighted ALRs up to positive affine transformations (Theorem~\ref{thm:weighted-alr}).

Third, we extend ALR and its weighted versions to many-to-many matching markets, allowing assortativeness to be measured both at intermediate stages of partner selection and in other economic settings where agents interact with multiple counterparts.\footnote{For real-world examples of many-to-many matching markets,
see the ``Related literature'' section.}
The random-matching benchmark underlying the
original ALR is designed for one-to-one markets and does not
directly account for individuals having multiple partners.
We introduce a benchmark based on {\it simple-graph random
matching}, in which men and women are matched completely at
random subject to two restrictions: each individual's number
of partners is held fixed, and no man--woman pair is matched
more than once. 
Although this benchmark has a natural interpretation, we are
not aware of a simple explicit formula for computing it.
To address this computational obstacle, we introduce a second
benchmark based on {\it configuration random matching}.
This procedure replaces each individual with one copy for each
observed partner and matches these copies one-to-one uniformly
at random, allowing repeated matches between the original
individuals. The resulting benchmark has a simple explicit
formula and provides an approximation to the first.
We establish a finite-market bound on the difference between
the two benchmarks (Proposition~\ref{prop:many-finite-bound}).
We also establish a large-market result: under suitable
conditions, this difference divided by the total number of
matches converges to zero as the market grows (Proposition~\ref{prop:many-asymptotic-bound}).
Finally, because configuration random matching amounts to
one-to-one matching between copies, we can directly apply our
earlier axiomatization to the corresponding many-to-many
measure (Corollary~\ref{thm:many-cm-alr}).

Fourth, we apply the many-to-many ALRs to data from IBJ, a Japanese marriage agency platform \citep{otani2025nonparametric}. The platform follows each pair through four stages of partner selection: application, pre-relationship, serious relationship, and proposal (engagement). A single market therefore yields a many-to-many matching at the early stages and an almost one-to-one matching at the final stage. Other data capture only one end of this process: marriage registers record the final one-to-one matching, and online dating data record the initial many-to-many contacts. The platform verifies key profile attributes, including age and income, against official documents. These profile attributes are available in a single database.
Previous studies measure sorting on different attributes in different data sources, such as income in tax records \citep{chiappori2026assortative} and anthropometric attributes such as height in household surveys \citep{chiappori2012fatter}; our data allow us to compare assortativeness across attributes within the same platform and period, relative to a common random-matching benchmark. We focus on three multi-valued attributes that are widely used in the literature on assortative matching: age, income, and height.

Two findings stand out. First, assortativeness is substantially stronger for age than for income or height. Second, age assortativeness increases particularly in the transitions from application to pre-relationship and from pre-relationship to serious relationship, with little further change at the proposal stage. These patterns show that much of the observed increase in age assortativeness occurs during the earlier transitions toward commitment.

Because our axiomatically characterized measures require only data on who matches with whom and on agents' attributes, without specifying preferences, transfers, or an equilibrium concept, they apply well beyond marriage markets, to any market in which many-to-many contacts narrow down to final matches.

Throughout this paper, we use the aggregate likelihood ratio and its variants as measures of assortativeness. \citet{chiappori2025changes} also study the \emph{odds ratio} and the \emph{normalized trace}. In our supplementary results, we identify errors in their axiomatizations of these two indices, develop and axiomatize multi-type odds ratios (Theorem~\ref{thm:multi-odds}), and provide a linear representation of weighted normalized-trace indices (Theorem~\ref{thm:linear-score}).

The paper proceeds as follows. Section~\ref{subsec:alr-revisited} presents the counterexample, recovers the axiomatization of ALR, and axiomatizes weighted ALRs. Section~\ref{subsec:many-to-many-alr} develops the many-to-many measures, establishes approximation bounds, and provides the corresponding axiomatization. Section~\ref{sec:application} presents the empirical application. Section~\ref{sec:conclusion} concludes. The appendices also contain the supplementary results on the odds ratio and normalized trace.

\subsection*{Related literature}
\label{subsec:related-literature}

The aggregate likelihood ratio (ALR) has been used in the literature
to measure assortativeness in marriage markets.
\citet{greenwood2014marry} use this measure to document increasing
educational assortativeness among U.S.\ married couples between
1960 and 2005.
\citet{eika2019educational} use type-specific likelihood ratios and
ALR to study educational assortative mating across five countries
and examine its implications for household income inequality.
\citet{chiappori2025changes} compare ALR with alternative indices
and show that conclusions about changes in educational sorting
depend on the measure used.
\citet{choi2025sorting} apply similarity-weighted ALRs and find
stronger educational sorting among more recent U.S.\ birth cohorts.
These studies examine one-to-one marriage markets.
We extend the ALR framework to many-to-many matching markets
as well, allowing us to measure assortativeness at earlier stages
of partner selection, when individuals interact with multiple
potential partners.

To our knowledge, this paper is the first in the economics literature to introduce axiomatically founded measures of assortativeness for two-sided many-to-many matching markets and apply them to real-world data.\footnote{Existing measures applicable to many-to-many matching markets include the assortativity coefficients of \citet{newman2003mixing}, which capture normalized excess same-type matching or correlations between partners' attributes. \citet{pelechrinis2016va} propose the VA-index, based on standardized differences between observed similarity and similarity under randomized-network benchmarks. Neither study provides an axiomatic characterization.}
Our measures retain ALR's interpretation relative to a random-matching benchmark and reduce to the corresponding one-to-one measures when each individual has at most one partner.
This provides a common framework for measuring assortativeness across successive stages of partner selection, from initial contacts to final matches.

Our measures also apply to sorting in other many-to-many matching markets. In labor markets, applications and interviews form a many-to-many network between job seekers and vacancies \citep{banfi2019high, wolthoff2018applications}, sorting is already present at the application stage \citep{banfi2022sorting}, and the way interviews narrow this network shapes the final one-to-one match \citep{lee2017interviewing, echenique2022top}. In industrial organization and international trade, firms form many-to-many buyer--seller relationships, such as those between automotive suppliers and assemblers \citep{fox2018estimating} and between exporters and importers \citep{bernard2018two}. Studies document sorting on capability between exporters and importers \citep{sugita2023assortative} and on wages in production networks \citep{demir2024ring}, each with a market-specific method or a structural matching model. Our measures require no matching model and therefore allow sorting to be compared across markets and stages of selection on a common scale. 

Our empirical analysis complements research on assortative matching in marriage. 
This research measures sorting among married couples, typically on one attribute at a time, such as age \citep{choo2006marries}, education \citep{chiappori2017partner, eika2019educational}, and income \citep{chiappori2026assortative}. Multidimensional approaches study sorting on several attributes jointly, including age, height, body mass index, wages, and education \citep{chiappori2012fatter, dupuy2014personality, ciscato2020role, ciscato2020like, chiappori2024analyzing}. These studies examine observed unions; our data reveal earlier stages of partner selection. Research on online dating also examines early choices: \citet{hitsch2010matching} analyze initial contacts and subsequent exchanges of contact information, but do not observe the resulting offline relationships.\footnote{Other studies measure marital sorting in labor force surveys \citep{chiappori2020changes}, household panels \citep{fremeaux2024assortative}, marriage registers \citep{hoehn2023changes}, administrative registers \citep{almar2024optimal}, biobanks \citep{abdellaoui2023trading}, and historical genealogies \citep{shiue2022marriage}. Other studies examine partner choice in online dating \citep{hitsch2010makes, ong2015income, bapna2016one, egebark2021brains, buyukeren2026impact} and speed dating \citep{fisman2006gender, belot2013dating}. \citet{lee2016effect} uses verified information on spouses to compare marital sorting between couples who met through online dating and those who met through other search methods; these data contain the resulting marriages but not the earlier stages of partner selection.} 
Our data follow pairs through successive stages, from initial applications through nonexclusive and exclusive relationships to engagement. This allows us to examine when assortativeness becomes more pronounced as relationships progress toward marriage.

\section{Axiomatization of the aggregate likelihood ratio revisited}
\label{subsec:alr-revisited}

This section %
proceeds as follows. 
After introducing the matching environment and axioms in Section~\ref{subsec:matching-index-axioms}, we present a counterexample to \cites{chiappori2025changes} characterization of the aggregate likelihood ratio in Section~\ref{subsec:counterexample}. Section~\ref{subsec:recovering-alr} completes our first contribution by providing two modifications that recover an axiomatization of ALR. Section~\ref{subsec:weighted-alr} presents an axiomatization of similarity-weighted ALRs.

\subsection{Matching, indices, and axioms}
\label{subsec:matching-index-axioms}

Fix an integer $n\geq2$ and let $T=\{1,\ldots,n\}$ be the set of types, common to men and women. A matching matrix is an $n\times n$ matrix
\[
M=(m_{ij})_{i,j\in T}\in\R_{\geq0}^{n\times n}\setminus\{0\},
\]
where $m_{ij}$ is the number, or mass, of matches between type-$i$ men and type-$j$ women. For each $M$, write
\[
|M|:=\sum_{i=1}^n\sum_{j=1}^n m_{ij},
\qquad
x_i(M):=\sum_{j=1}^n m_{ij},
\qquad
y_i(M):=\sum_{j=1}^n m_{ji}.
\]
Thus, $|M|$ is the total mass of matched pairs, while $x_i(M)$ and $y_i(M)$ are the respective masses of type-$i$ men and women. Define the observed mass of same-type pairs by
\begin{equation}\label{eq:sigma}
\sigma(M):=\sum_{i=1}^n m_{ii}.
\end{equation}
The corresponding expected mass under random matching, holding the row and column marginals fixed, is
\begin{equation}\label{eq:rho}
\rho(M):=\sum_{i=1}^n\frac{x_i(M)y_i(M)}{|M|}.
\end{equation}
Because ALR divides by this benchmark, we use the domain
\[
\Mn:=\bigl\{M\in\R_{\geq0}^{n\times n}\setminus\{0\}:\rho(M)>0\bigr\}.
\]
A matching matrix is \emph{positive} if all its entries are positive. Let $\Mnpos$ denote the set of positive matching matrices, and let $\one$ denote the $n\times n$ all-one matrix.

The \textbf{aggregate likelihood ratio} (ALR) is
\begin{equation}\label{eq:multitype-alr}
I_L^n(M):=\frac{\sigma(M)}{\rho(M)}
=\frac{\sum_{i=1}^n m_{ii}/|M|}
{\sum_{i=1}^n(x_i(M)/|M|)(y_i(M)/|M|)}
\qquad(M\in\Mn).
\end{equation}
ALR compares observed same-type matching with the amount predicted by random matching with the same marginals.\footnote{\citet{ciscato2020like} introduced a related measure called the {\it homogamy rate} for type pair $(i,j)$. It measures the number of observed $(i,j)$-pairs relative to the expected number of such pairs under random matching.} A value of one means that the two amounts coincide; values above one indicate more same-type matching than random matching predicts, and values below one indicate less.
\begin{remark}
Alongside ALR, \citet{chiappori2025changes} examine alternative indices of assortativeness, most notably the odds ratio. We focus on ALR for two reasons. First, its interpretation relative to random matching provides an intuitive benchmark for comparing assortativeness across attributes.
Second, ALR extends naturally to multitype markets, whereas the odds ratio has no natural extension \citep[Section VI.A]{chiappori2025changes}. We also extend ALR to many-to-many matching markets; see Section \ref{subsec:many-to-many-alr}.

\citet[pp.~3073--3074]{chiappori2025changes} argue that marginal
independence, a property that ALR does not satisfy, is particularly
relevant when the marginal distributions shift substantially. In their
U.S.\ data, the largest change in the share of women in any of the five
education groups was 14.5 percentage points between the 1950s and 1970s
birth cohorts. In our engagement sample, the largest change in the
share of women in any of the five age groups was 1.7 percentage points
between 2021 and 2024 (Table~\ref{tb:comparison_chiappori}(a)
in Section~\ref{subsubsec:results_comparison}).
The relative stability of women’s age-group shares therefore suggests that ALR’s sensitivity to changes in the female age distribution is less of a concern when we compare engaged couples over time.
\end{remark}

Unless otherwise stated, an index is a function $I:\Mn\to\R_{\geq0}$. For a permutation $\pi$ of $T$, let $M^\pi$ be the matrix obtained by simultaneously relabeling men's and women's types:
\[
(M^\pi)_{ij}=m_{\pi(i)\pi(j)}
\qquad(i,j\in T).
\]
Write $M^{\mathsf T}$ for the transpose of $M$. For distinct types $p,q\in T$ and $\eps\in\R$, define the elementary change $M^{pq,\eps}$ by
\begin{equation}\label{eq:elementary-change-def}
(M^{pq,\eps})_{kl}=
\begin{cases}
m_{kl}+\eps, & (k,l)=(p,p)\text{ or }(k,l)=(q,q),\\
m_{kl}-\eps, & (k,l)=(p,q)\text{ or }(k,l)=(q,p),\\
m_{kl}, & \text{otherwise.}
\end{cases}
\end{equation}
This change preserves all row and column sums and increases the mass of same-type pairs by $2\eps$.

We use the following axioms.
\begin{itemize}
\item \textbf{Scale Invariance.} For every $M\in\Mn$ and every $\lambda>0$, $I(\lambda M)=I(M)$.

\item \textbf{Side Invariance.} For every $M\in\Mn$, $I(M^{\mathsf T})=I(M)$.

\item \textbf{Type Invariance.} For every $M\in\Mn$ and every permutation $\pi$ of $T$, $I(M^\pi)=I(M)$.

\item \textbf{Marginal Monotonicity.} For every $M\in\Mnpos$, every pair of distinct types $p,q\in T$, and every $\eps\in\R$ such that $M^{pq,\eps}\in\Mnpos$,
\[
I(M^{pq,\eps})>I(M)\quad\Longleftrightarrow\quad\eps>0.
\]

\item \textbf{Random Decomposability.} For every $M,M'\in\Mnpos$,
\begin{equation}\label{eq:multi-rd}
I(M+M')=\frac{\rho(M)}{\rho(M+M')}I(M)
+\frac{\rho(M')}{\rho(M+M')}I(M').
\end{equation}
\end{itemize}

The three invariance axioms require the index to be unaffected by a change in the unit of measurement, an exchange of the two sides of the market, or a simultaneous relabeling of types. Marginal Monotonicity requires assortativeness to increase when an elementary change replaces different-type pairs with same-type pairs while preserving the marginals. Random Decomposability specifies how the index of a pooled population relates to the indices of its component populations. 

\subsection{Counterexample to Chiappori et al.'s (2025) theorem}
\label{subsec:counterexample}

For two-type markets, \citet{chiappori2025changes} state the following characterization.
\begin{theoremCDMZthree}
For $n=2$, an index satisfies Scale Invariance, Side Invariance, Type Invariance, Marginal Monotonicity, and Random Decomposability if and only if it is a positive multiple of the aggregate likelihood ratio.
\end{theoremCDMZthree}

To see why this statement fails, specialize the notation of Section~\ref{subsec:matching-index-axioms} to $n=2$ and write
\[
M=\begin{pmatrix}a&b\\c&d\end{pmatrix}\in\mathcal M^2.
\]
Then the aggregate likelihood ratio is given by 
\begin{equation}
I_L^2(M):=\frac{a+d}{\rho(M)}
\qquad(M\in\mathcal M^2).
\end{equation}
Consider the following variation: 
\begin{equation}\label{eq:two-counter-index}
I_L'(M):=\frac{a+d+\frac12 b+\frac12 c}{\rho(M)}
\qquad(M\in\mathcal M^2).
\end{equation}
This index satisfies all five axioms in Theorem CDMZ3. Both its numerator and denominator are homogeneous of degree one and invariant under transposition and simultaneous relabeling of the two types. Hence the three invariance axioms hold. Its numerator is additive, so Random Decomposability holds. Finally, the elementary change $M^{12,\eps}$ leaves $\rho(M)$ unchanged and increases the numerator by $\eps$, giving
\[
I_L'(M^{12,\eps})-I_L'(M)=\frac{\eps}{\rho(M)}.
\]
Thus, Marginal Monotonicity also holds whenever the two matrices are positive.

Nevertheless, $I_L'$ is not a positive multiple of $I_L^2$. Suppose, to the contrary, that it were. Then, for every $M,M'\in\mathcal M^2$,
\begin{equation}\label{eq:inequality-equivalence}
I_L^2(M)>I_L^2(M')\quad\Longleftrightarrow\quad I_L'(M)>I_L'(M').
\end{equation}
Consider
\[
M=\begin{pmatrix}1&1\\1&1\end{pmatrix},
\qquad
M'=\begin{pmatrix}1&1\\3&2\end{pmatrix}.
\]
A direct calculation gives
\[
I_L^2(M)=1>\frac{21}{23}=I_L^2(M'),
\qquad
I_L'(M)=\frac32<\frac{35}{23}=I_L'(M'),
\]
a contradiction to (\ref{eq:inequality-equivalence}).

\subsection{Recovering the axiomatization of ALR}
\label{subsec:recovering-alr}

We now complete the first contribution by presenting two modifications of the original axiomatization, both valid for any finite number of types. Section~\ref{subsec:axiomatization-alr} retains Random Decomposability and adds Full Heterogamy and Continuity to characterize positive multiples of ALR. Section~\ref{subsec:alr-without-rd} instead uses Random-Benchmark Monotonicity to characterize increasing transformations of ALR and, with additional restrictions on the cardinal scale, positive affine transformations.

\subsubsection{Axiomatization of ALR based on random decomposability}\label{subsec:axiomatization-alr}

We first identify all indices satisfying Scale Invariance, Marginal Monotonicity, and Random Decomposability. %

\begin{proposition} \label{prop:core-representation}
An index $I:\Mn\to\R_{\geq 0}$ satisfies Scale Invariance, Marginal Monotonicity, and Random Decomposability if and only if the following two conditions hold.
\begin{enumerate}[label=(\roman*)]
\item There exists a coefficient matrix $\Gamma=(\gamma_{ij})_{i,j\in T}\in \R_{\geq 0}^{n\times n}$ such that, for every $M\in\Mnpos$,
\begin{equation}\label{eq:gamma-representation}
I(M)=\frac{\sum_{i=1}^n\sum_{j=1}^n \gamma_{ij}m_{ij}}{\rho(M)},
\end{equation}
and, for every pair of distinct types $p,q\in T$,
\begin{equation}\label{eq:gamma-monotone}
\gamma_{pp}+\gamma_{qq}>\gamma_{pq}+\gamma_{qp}.
\end{equation}
\item On $\Mn\setminus\Mnpos$, $I$ is an arbitrary index that satisfies Scale Invariance.
\end{enumerate}
\end{proposition}

\begin{proof}
First suppose that conditions (i) and (ii) hold. Scale Invariance holds on $\Mn\setminus\Mnpos$ by condition (ii). On $\Mnpos$, both the numerator in \eqref{eq:gamma-representation} and $\rho(M)$ are multiplied by $\lambda$ when $M$ is replaced by $\lambda M$, so Scale Invariance holds as well.

Let $A_\Gamma(M):=\sum_{i,j}\gamma_{ij}m_{ij}$. If $M,M'\in\Mnpos$, then $A_\Gamma(M+M')=A_\Gamma(M)+A_\Gamma(M')$. Therefore,
\[
\frac{\rho(M)}{\rho(M+M')}I(M)+\frac{\rho(M')}{\rho(M+M')}I(M')
=\frac{A_\Gamma(M)+A_\Gamma(M')}{\rho(M+M')}
=I(M+M'),
\]
so Random Decomposability holds.

It remains to verify Marginal Monotonicity. Fix $M\in\Mnpos$, distinct $p,q\in T$, and $\eps\in\R$ such that $M^{pq,\eps}\in\Mnpos$. The operation defining $M^{pq,\eps}$ leaves every row and column sum unchanged, and hence $\rho(M^{pq,\eps})=\rho(M)$. The numerator changes by
\[
A_\Gamma(M^{pq,\eps})-A_\Gamma(M)
=\eps(\gamma_{pp}+\gamma_{qq}-\gamma_{pq}-\gamma_{qp}).
\]
By \eqref{eq:gamma-monotone}, this expression is positive if and only if $\eps>0$. Thus Marginal Monotonicity holds.

Conversely, suppose that $I$ satisfies the three axioms. Define
\[
F(M):=\rho(M)I(M)\qquad (M\in\Mnpos).
\]
Since $\rho(M)>0$ and $I(M)\geq 0$, we have $F(M)\geq 0$ for all $M\in\Mnpos$. Random Decomposability implies additivity:
\[
F(M+M')=F(M)+F(M')\qquad(M,M'\in\Mnpos).
\]
Moreover, since $\rho(\lambda M)=\lambda\rho(M)$ and $I(\lambda M)=I(M)$ (which follows from Scale Invariance),
\[
F(\lambda M)=\lambda F(M)\qquad(M\in\Mnpos,\ \lambda>0).
\]
Thus $F$ is additive and positively homogeneous of degree one on $\Mnpos$.

We next extend $F$ to a linear functional on all of $\R^{n\times n}$. For an arbitrary matrix $X\in\R^{n\times n}$, choose $\lambda>0$ large enough that $X+\lambda\one\in\Mnpos$, and define
\[
\overline F(X):=F(X+\lambda\one)-F(\lambda\one).
\]
This definition is independent of the admissible choice of $\lambda$. If $\lambda'>\lambda$ is another admissible choice, then by additivity,
\[
F(X+\lambda'\one)-F(\lambda'\one)
=F(X+\lambda\one)+F((\lambda'-\lambda)\one)-F(\lambda\one)-F((\lambda'-\lambda)\one),
\]
which is $F(X+\lambda\one)-F(\lambda\one)$.

The extension $\overline F$ is additive. For arbitrary $X,Y\in\R^{n\times n}$, choose $\lambda>0$ large enough that $X+\lambda\one$, $Y+\lambda\one$, and $X+Y+2\lambda\one$ are all positive. Then
\[
\overline F(X+Y)=F(X+Y+2\lambda\one)-F(2\lambda\one)
=F(X+\lambda\one)+F(Y+\lambda\one)-2F(\lambda\one)
=\overline F(X)+\overline F(Y).
\]
For $c>0$,
\[
\overline F(cX)=F(cX+c\lambda\one)-F(c\lambda\one)
=cF(X+\lambda\one)-cF(\lambda\one)=c\overline F(X).
\]
This equation implies that $\overline F$ satisfies homogeneity for all $c>0$. We prove that $\overline F$ satisfies homogeneity for $c=0$ and $c<0$ as well. Since
$$
\bar{F}(0)=\bar{F}(0)+\bar{F}(0),
$$
we have $\bar{F}(0)=0$. Moreover,
$$
0=\bar{F}(0)=\bar{F}(X+(-X))=\bar{F}(X)+\bar{F}(-X),
$$
which implies
$$
\bar{F}(-X)=-\bar{F}(X) .
$$
Now, for $c<0$, write $c=-d$, where $d>0$. Then
$$
\bar{F}(c X)=\bar{F}(-d X)=-\bar{F}(d X)=-d \bar{F}(X)=c \bar{F}(X),
$$
where the penultimate equality uses the positive homogeneity already established for $d>0$.

It follows that $\overline F$ is a linear functional on the finite-dimensional vector space $\R^{n\times n}$. Hence there are coefficients $(\gamma_{ij})_{i,j\in T}\in\R^{n\times n}$ such that
\[
F(M)=\sum_{i=1}^n\sum_{j=1}^n\gamma_{ij}m_{ij}
\qquad(M\in\Mnpos).
\]
Because $F(M)\geq 0$ for every positive $M$, all coefficients are nonnegative. Indeed, if some $\gamma_{pq}<0$, then by taking $m_{pq}$ arbitrarily large and keeping all other entries positive and fixed, the right-hand side becomes negative, a contradiction.

Finally, Marginal Monotonicity implies \eqref{eq:gamma-monotone}. Fix distinct $p,q$ and choose $M\in\Mnpos$ and $\eps>0$ small enough that $M^{pq,\eps}\in\Mnpos$. Since $\rho(M^{pq,\eps})=\rho(M)$ and $I(M^{pq,\eps})>I(M)$, we must have
\[
\eps(\gamma_{pp}+\gamma_{qq}-\gamma_{pq}-\gamma_{qp})>0.
\]
Thus $\gamma_{pp}+\gamma_{qq}>\gamma_{pq}+\gamma_{qp}$. Scale Invariance on $\Mn\setminus\Mnpos$ is inherited directly from the axiom. This proves the proposition.
\end{proof}

Proposition \ref{prop:core-representation} indicates that, to axiomatize the aggregate likelihood ratio, we need to add additional axioms that rule out dependence on off-diagonal entries, and fill in the gap between the functional forms on the positive matrices and other matrices. We introduce axioms that serve these roles. 
\begin{itemize}
\item \textbf{Full Heterogamy.} Let
\[
\mathcal H^n:=\bigl\{H=(h_{ij})\in\Mn:
h_{ii}=0\text{ for all }i\in T,\;
h_{ij}>0\text{ for all }i\neq j\bigr\}.
\]
For every $M\in\Mnpos$ and every $H\in\mathcal H^n$, $I(M)\geq I(H)$.

\item \textbf{Continuity.} The function $I:\Mn\to\R_{\geq0}$ is continuous.
\end{itemize}

Full Heterogamy requires any positive matching matrix to be weakly more assortative than a matrix with no same-type pairs and positive mass in every off-diagonal cell. Continuity rules out jumps in the index when the matching matrix changes slightly.

\begin{theorem} \label{thm:multitype-alr}
An index $I:\Mn\to\R_{\geq 0}$ satisfies Scale Invariance, Type Invariance, Marginal Monotonicity, Random Decomposability, Full Heterogamy, and Continuity if and only if there exists $\alpha>0$ such that
\[
I(M)=\alpha I_L^n(M)
\qquad\text{for all }M\in\Mn.
\]
\end{theorem}

\begin{proof}
Suppose first that $I(M)=\alpha I_L^n(M)$ for some $\alpha>0$. By definition, $I(M)=\alpha\sigma(M)/\rho(M)$. Scale Invariance follows because both $\sigma(M)$ and $\rho(M)$ are multiplied by $\lambda$ when $M$ is multiplied by $\lambda>0$. Type Invariance follows because a simultaneous relabeling of men's and women's types leaves both $\sigma(M)$ and $\rho(M)$ unchanged.

For Marginal Monotonicity, fix $M\in\Mnpos$, distinct $p,q\in T$, and $\eps\in\R$ such that $M^{pq,\eps}\in\Mnpos$. The transformation keeps all row and column sums fixed, so $\rho(M^{pq,\eps})=\rho(M)$. Also $\sigma(M^{pq,\eps})=\sigma(M)+2\eps$. Hence
\[
I(M^{pq,\eps})-I(M)=\frac{2\alpha\eps}{\rho(M)},
\]
which is positive if and only if $\eps>0$.

Random Decomposability follows from the linearity of $\sigma$:
\[
\frac{\rho(M)}{\rho(M+M')}I(M)+\frac{\rho(M')}{\rho(M+M')}I(M')
=\frac{\alpha \sigma(M)+\alpha \sigma(M')}{\rho(M+M')}
=I(M+M').
\]
Continuity is immediate because $\sigma$ and $\rho$ are continuous and $\rho(M)>0$ on $\Mn$. Finally, we verify Full Heterogamy. If $H\in\mathcal H^n$, then $\sigma(H)=0$ and hence $I(H)=0$, while $I(M)>0$ for every $M\in\Mnpos$. Thus Full Heterogamy holds.

Conversely, suppose that $I$ satisfies the six axioms. By Proposition \ref{prop:core-representation}, on $\Mnpos$ there is a nonnegative coefficient matrix $\Gamma=(\gamma_{ij})$ such that
\[
I(M)=\frac{\sum_{i,j}\gamma_{ij}m_{ij}}{\rho(M)}.
\]
Type Invariance implies that the coefficient attached to a diagonal cell is common across all diagonal cells, and the coefficient attached to an off-diagonal cell is common across all off-diagonal cells. Hence there exist $\alpha,\beta\geq 0$ such that
\[
\gamma_{ij}=\begin{cases}
\alpha, & i=j,\\
\beta, & i\neq j.
\end{cases}
\]
The coefficient inequality \eqref{eq:gamma-monotone} gives $2\alpha>2\beta$, so $\alpha>\beta$. Therefore, for every $M\in\Mnpos$,
\begin{equation}\label{eq:alpha-beta-multitype}
I(M)=\frac{\alpha\sum_i m_{ii}+\beta\sum_{i\neq j}m_{ij}}{\rho(M)}.
\end{equation}

We show that $\beta=0$. By continuity, \eqref{eq:alpha-beta-multitype} extends to every matrix in $\mathcal H^n$. Indeed, if $H\in\mathcal H^n$ and $H_k:=H+k^{-1}\one$ ($k=1, 2, \dots$), then $H_k\in\Mnpos$ and $H_k\to H$ as $k\rightarrow \infty$, so
\[
I(H)=\frac{\beta\sum_{i\neq j}h_{ij}}{\rho(H)}.
\]
Suppose, toward a contradiction, that $\beta>0$. For $t>0$, define $H^t=(h_{ij}^t)\in\mathcal H^n$ by
\[
h^t_{12}=t,
\qquad
h^t_{ij}=1\quad\text{for all }i\neq j\text{ with }(i,j)\neq(1,2),
\qquad
h^t_{ii}=0\quad\text{for all }i.
\]
Then $\sum_{i\neq j}h_{ij}^t=t+n(n-1)-1$, while
\[
\rho(H^t)=\frac{2(n-1)(t+n-2)+(n-2)(n-1)^2}{t+n(n-1)-1}\longrightarrow 2(n-1)
\quad\text{as }t\to\infty.
\]
Thus $I(H^t)\to\infty$. This contradicts Full Heterogamy, because for the all-one matrix $\one\in\Mnpos$, the value $I(\one)$ is finite and Full Heterogamy requires $I(\one)\geq I(H^t)$ for all $t>0$. Hence $\beta=0$.

Since $\alpha>\beta=0$, we have $\alpha>0$. Equation \eqref{eq:alpha-beta-multitype} becomes $I(M)=\alpha I_L^n(M)$ on $\Mnpos$. Finally, for any $M\in\Mn$, the sequence $M_k:=M+k^{-1}\one$ ($k=1, 2, \dots$) is positive and converges to $M$. By continuity of $I$ and $I_L^n$ on $\Mn$,
\[
I(M)=\lim_{k\to\infty}I(M_k)=\alpha\lim_{k\to\infty}I_L^n(M_k)=\alpha I_L^n(M).
\]
\end{proof}

Thus, Full Heterogamy and Continuity recover the intended characterization when added to the original axioms. Side Invariance need not be imposed separately: it follows from the representation in Theorem~\ref{thm:multitype-alr}.

\subsubsection{Axiomatization of ALR based on a new monotonicity axiom}
\label{subsec:alr-without-rd}

Theorem~\ref{thm:multitype-alr} retains Random Decomposability. Although this axiom plays a useful mathematical role, its economic content is less transparent than that of Scale Invariance or Marginal Monotonicity.
Indeed, \citet{chiappori2025changes} introduce Random Decomposability as a ``characterization axiom'' for the ALR rather than as a basic desirable property of assortative matching measures. 
In particular, as they point out, the two coefficients in equation~\eqref{eq:multi-rd} need not sum to one. Thus, the equation does not generally express the index of a pooled population as a weighted average of the indices of its components.\footnote{\citet{chiappori2025changes} introduce a weaker variant, Marginal Random Decomposability, under which the coefficients sum to one. We return to this axiom in equation~\eqref{eq:marginal-rd} below.}

Our second modification replaces Random Decomposability with a monotonicity axiom that compares the observed same-type mass $\sigma(M)$ with its random-matching benchmark $\rho(M)$.
\begin{itemize}

\item \textbf{Random-Benchmark Monotonicity.}
For every $M,M'\in\Mn$, if
\[
\rho(M)\leq \rho(M')
\qquad\text{and}\qquad
\sigma(M)\geq \sigma(M'),
\]
then
\[
I(M)\geq I(M').
\]
\end{itemize}
Random-Benchmark Monotonicity requires the index to move in the unambiguous direction when the two quantities move in opposite directions. If a matching table contains weakly more observed same-type matches while its random-matching benchmark is weakly lower, then it is regarded as weakly more assortative. Both changes strengthen the evidence for assortative matching: more same-type matches are observed, while fewer such matches would arise simply from the marginal type distributions.

The axiom directly describes how changes in the marginal type distributions affect measured assortativeness. It thereby provides a useful counterpart to Marginal Independence in \citet{chiappori2025changes}, which treats proportional rescalings of a row or column as assortatively neutral.
The two axioms make the contrast between ALR and the odds ratio more transparent by formalizing different views of how changing marginals should affect assortativeness.\footnote{Marginal Independence together with other standard axioms axiomatizes the odds ratio. See Section \ref{sect-odds} for details.}

The next proposition identifies the class of indices satisfying Random-Benchmark Monotonicity and two standard axioms.
Let
\[
\mathcal{R}:=\bigl\{I_L^n(M):M\in\Mn\bigr\}
\]
denote the range of ALR.

\begin{proposition}
\label{prop-monotonic-ALR}
An index $I:\Mn\to\R_{\geq0}$ satisfies Scale Invariance, Marginal Monotonicity, and Random-Benchmark Monotonicity if and only if there exists a strictly increasing function $f:\mathcal{R}\to\R_{\geq0}$ such that
\[
I(M)=f\bigl(I_L^n(M)\bigr)
\quad\text{for every }M\in\Mn.
\]
\end{proposition}

\begin{proof}
Write
\[
L(M):=I_L^n(M)=\frac{\sigma(M)}{\rho(M)}.
\]
We first prove the ``if'' direction. Suppose that
\[
I(M)=f(L(M))
\]
for a strictly increasing function $f:\mathcal{R}\to\R_{\geq0}$. Since both $\sigma$ and $\rho$ are homogeneous of degree one, $L(\lambda M)=L(M)$ for every $\lambda>0$. Hence Scale Invariance holds.

To verify Random-Benchmark Monotonicity, suppose that
\[
\rho(M)\leq\rho(M')
\quad\text{and}\quad
\sigma(M)\geq \sigma(M').
\]
Because $\sigma(M),\sigma(M')\geq0$ and $\rho(M),\rho(M')>0$,
\[
L(M)=\frac{\sigma(M)}{\rho(M)}
\geq
\frac{\sigma(M')}{\rho(M)}
\geq
\frac{\sigma(M')}{\rho(M')}
=L(M').
\]
Since $f$ is increasing, $I(M)\geq I(M')$.

For Marginal Monotonicity, fix $M\in\Mnpos$, distinct types $p,q\in T$, and $\varepsilon\in\R$ such that $M^{pq,\varepsilon}\in\Mnpos$. The elementary change preserves all row and column marginals and therefore leaves $\rho$ unchanged, while it changes the diagonal mass by $2\varepsilon$. Hence
\[
L(M^{pq,\varepsilon})
=
L(M)+\frac{2\varepsilon}{\rho(M)}.
\]
Since $f$ is strictly increasing,
\[
I(M^{pq,\varepsilon})>I(M)
\quad\Longleftrightarrow\quad
\varepsilon>0.
\]
Thus Marginal Monotonicity holds.

We next prove the ``only if'' direction. Suppose that $I$ satisfies the three axioms. We first show that $I(M)$ depends on $M$ only through $L(M)$. Take $M,M'\in\Mn$ such that
\[
L(M)=L(M').
\]
Set
\[
c:=\frac{\rho(M')}{\rho(M)}>0.
\]
By homogeneity of $\sigma$ and $\rho$,
\[
\begin{aligned}
\rho(cM)&=\rho(M'),\\
\sigma(cM)&=c\sigma(M)
=\rho(M')L(M)\\
&=\rho(M')L(M')
=\sigma(M').
\end{aligned}
\]
Random-Benchmark Monotonicity applied in both directions gives
\[
I(cM)=I(M').
\]
Scale Invariance gives $I(cM)=I(M)$, and therefore
\[
L(M)=L(M')
\quad\Longrightarrow\quad
I(M)=I(M').
\]
It follows that there exists a function $f:\mathcal{R}\to\R_{\geq0}$ such that
\[
I(M)=f(L(M))
\quad\text{for every }M\in\Mn.
\]

We show that $f$ is nondecreasing. Take $x,y\in\mathcal{R}$ with $x\leq y$, and choose $M,M'\in\Mn$ such that $L(M)=x$ and $L(M')=y$. Define
\[
\widehat{M}:=\frac{M}{\rho(M)}
\quad\text{and}\quad
\widehat{M}':=\frac{M'}{\rho(M')}.
\]
Then
\[
\rho(\widehat{M})=\rho(\widehat{M}')=1,
\quad
\sigma(\widehat{M})=x,
\quad\text{and}\quad
\sigma(\widehat{M}')=y.
\]
Random-Benchmark Monotonicity therefore gives
\[
I(\widehat{M}')\geq I(\widehat{M}).
\]
By Scale Invariance,
\[
f(y)=I(M')=I(\widehat{M}')
\geq I(\widehat{M})=I(M)=f(x).
\]
Thus $f$ is nondecreasing on $\mathcal{R}$.

It remains to prove that $f$ is strictly increasing. Let
\[
\mathcal{R}_{+}:=\bigl\{L(M):M\in\Mnpos\bigr\}.
\]
The set $\Mnpos$ is convex and hence connected, and $L$ is continuous. Therefore, $\mathcal{R}_{+}$ is an interval. Moreover, it is open. Indeed, take any $M\in\Mnpos$ and any distinct types $p,q\in T$. For every sufficiently small positive or negative $\delta$, the matrix $M^{pq,\delta}$ remains positive, and
\[
L(M^{pq,\delta})
=
L(M)+\frac{2\delta}{\rho(M)}.
\]
Thus every point of $\mathcal{R}_{+}$ is an interior point.

The set $\mathcal{R}_{+}$ is also dense in $\mathcal{R}$. To see this, for every $M\in\Mn$, define
\[
M_k:=M+\frac{1}{k}\one
\quad (k=1,2,\ldots).
\]
Then $M_k\in\Mnpos$ and $M_k\to M$. By continuity of $L$,
\[
L(M_k)\to L(M).
\]
Hence $\mathcal{R}\subseteq\overline{\mathcal{R}_{+}}$, where the upper bar denotes the closure.

We first show that $f$ is strictly increasing on $\mathcal{R}_{+}$. Take $u,v\in\mathcal{R}_{+}$ with $u<v$, and choose $M\in\Mnpos$ such that $L(M)=u$. Fix distinct types $p,q\in T$. Since $M$ is positive, there exists $\varepsilon>0$ small enough that $M^{pq,\varepsilon}\in\Mnpos$ and
\[
u<L(M^{pq,\varepsilon})
=u+\frac{2\varepsilon}{\rho(M)}
<v.
\]
Marginal Monotonicity gives
\[
f\bigl(L(M^{pq,\varepsilon})\bigr)
=I(M^{pq,\varepsilon})
>I(M)=f(u).
\]
Since $f$ is nondecreasing and $L(M^{pq,\varepsilon})<v$, we obtain
\[
f(v)\geq f\bigl(L(M^{pq,\varepsilon})\bigr)>f(u).
\]

Finally, take arbitrary $x,y\in\mathcal{R}$ with $x<y$. Since $\mathcal{R}_{+}$ is an open interval and is dense in $\mathcal{R}$, there exist $u,v\in\mathcal{R}_{+}$ such that
\[
x<u<v<y.
\]
By the preceding argument and the nondecreasing property of $f$,
\[
f(y)\geq f(v)>f(u)\geq f(x).
\]
Thus $f$ is strictly increasing on $\mathcal{R}$.
\end{proof}

Proposition \ref{prop-monotonic-ALR} identifies the index only up to a strictly increasing transformation of ALR. This is natural because none of the three axioms imposes a restriction on the cardinal scale of the index. To restrict the admissible transformation to a positive affine one, we introduce an additional axiom.
\begin{itemize}
\item \textbf{Increment Invariance.}
Let $M,M'\in\Mnpos$ have the same row and column marginal vectors, that is,
\[
x_i(M)=x_i(M')
\qquad\text{and}\qquad
y_i(M)=y_i(M')
\qquad\text{for every }i\in T.
\]
Then, for every pair of distinct types $p,q\in T$ and every $\varepsilon>0$ such that
\[
M^{pq,\varepsilon},\,(M')^{pq,\varepsilon}
\in\Mnpos,
\]
we have
\[
I(M^{pq,\varepsilon})-I(M)
=
I((M')^{pq,\varepsilon})-I(M').
\]
\end{itemize}
When $M$ and $M'$ have the same row and column marginal vectors, they describe markets with the same type composition.
Thus, the same elementary change (i.e., increasing the mass of same-type pairs while fixing marginals) represents the same increase in same-type matching under the same opportunity structure. Increment Invariance requires this change to make the same contribution to measured assortativeness.

\begin{theorem}
\label{thm:alr-without-rd}
An index $I:\Mn\to\R_{\geq0}$ satisfies Scale Invariance, Random-Benchmark Monotonicity, Marginal Monotonicity, Increment Invariance, and Continuity if and only if there exist constants $\alpha>0$ and $\beta\geq0$ such that
\[
I(M)=\alpha I_L^n(M)+\beta
=\alpha\frac{\sigma(M)}{\rho(M)}+\beta
\qquad
\text{for every }M\in\Mn.
\]
\end{theorem}

\begin{proof}
For the ``if'' direction, suppose that $I(M)=\alpha I_L^n(M)+\beta$ for some $\alpha>0$ and $\beta\geq0$. By Proposition~\ref{prop-monotonic-ALR}, Scale Invariance, Random-Benchmark Monotonicity, and Marginal Monotonicity hold. If $M,M'\in\Mnpos$ have the same row and column marginal vectors, then $\rho(M)=\rho(M')$, and hence, whenever the elementary changes are feasible,
\[
I(M^{pq,\varepsilon})-I(M)
=
\frac{2\alpha\varepsilon}{\rho(M)}
=
\frac{2\alpha\varepsilon}{\rho(M')}
=
I((M')^{pq,\varepsilon})-I(M').
\]
Thus, Increment Invariance holds. Continuity is immediate because $\sigma$ and $\rho$ are continuous and $\rho>0$ on $\Mn$.

For the ``only if'' direction, Proposition~\ref{prop-monotonic-ALR} yields a strictly increasing function $f:\mathcal R\to\R_{\geq0}$ such that
\[
I(M)=f\bigl(I_L^n(M)\bigr)
\qquad
(M\in\Mn),
\]
where $\mathcal R$ is the range of $I_L^n$. Let
\[
\mathcal R_+:=\{I_L^n(M):M\in\Mnpos\}.
\]
As shown in the proof of Proposition~\ref{prop-monotonic-ALR}, $\mathcal R_+$ is an open interval.

Fix $x_0\in\mathcal R_+$ and choose $M^0\in\Mnpos$ with $I_L^n(M^0)=x_0$. By scaling $M^0$, assume $\rho(M^0)=1$. Fix distinct $p,q\in T$. For every $s$ in a sufficiently small open interval $U$ containing $x_0$, define
\[
M_s:=(M^0)^{pq,(s-x_0)/2}.
\]
Then $M_s\in\Mnpos$, all the matrices $M_s$ have the same row and column marginal vectors, and
\[
\rho(M_s)=1,
\qquad
I_L^n(M_s)=s.
\]
Hence, for $x,y,x+h,y+h\in U$ with $h>0$, Increment Invariance implies
\[
f(x+h)-f(x)=f(y+h)-f(y).
\]
Thus, locally, an increment of a given size changes $f$ by an amount independent of the starting point. Since $f$ is strictly increasing, the standard monotone Cauchy-equation argument implies that $f$ is affine on a neighborhood of $x_0$; see, for example, \citet[Chapter~1, Exercise~1, p.~21]{sahoo2011introduction}. Because $x_0$ was arbitrary and $\mathcal R_+$ is connected, the local affine representations agree on overlaps. Therefore, there exist constants $\alpha>0$ and $\beta$ such that
\begin{equation}
\label{eq:affine-alr-positive}
I(M)=\alpha I_L^n(M)+\beta
\qquad
(M\in\Mnpos).
\end{equation}

To extend the representation to all of $\Mn$, fix $M\in\Mn$ and set $M_k:=M+k^{-1}\one$. Then $M_k\in\Mnpos$ and $M_k\to M$. Continuity of $I$ and $I_L^n$ gives, from~\eqref{eq:affine-alr-positive},
\[
I(M)=\alpha I_L^n(M)+\beta.
\]
Finally, for the matrix $H^0$ with zero diagonal and all off-diagonal entries equal to one, $I_L^n(H^0)=0$. Hence $\beta=I(H^0)\geq0$. This completes the proof.
\end{proof}

Random-Benchmark Monotonicity also allows us to address a question raised by \citet{chiappori2025changes}. As noted at the beginning of this subsection, a limitation of Random Decomposability as a decomposition axiom is that it does not generally admit a weighted-average interpretation. To address this limitation, \citet{chiappori2025changes} introduce the following weaker variant.

\begin{itemize}
\item \textbf{Marginal Random Decomposability.}
For every $M,M'\in\Mnpos$ such that $M/|M|$ and $M'/|M'|$ have the same row and column marginal distributions, it holds that
\begin{equation}
\label{eq:marginal-rd}
I(M+M')
=
\frac{\rho(M)}{\rho(M+M')}I(M)
+
\frac{\rho(M')}{\rho(M+M')}I(M').
\end{equation}
\end{itemize}

Under the common-distribution restriction, the two coefficients in equation~\eqref{eq:marginal-rd} sum to one. In particular, if $M$ and $M'$ have the same row and column marginal vectors, then $|M|=|M'|$ and Marginal Random Decomposability reduces to
\[
I(M+M')
=
\frac{I(M)+I(M')}{2}.
\]
Under Scale Invariance, this can equivalently be written as
\[
I\left(\frac{M+M'}{2}\right)
=
\frac{I(M)+I(M')}{2}.
\]
Hence, on positive matching matrices with common marginals, Marginal Random Decomposability imposes the usual midpoint-averaging requirement. This interpretation is more transparent than that of Random Decomposability. Nevertheless, \citet{chiappori2025changes} do not provide an axiomatization of ALR based on this axiom, noting that ALR is not the unique measure satisfying it. We establish such an axiomatization by combining Marginal Random Decomposability with the other axioms below.

\begin{theorem}
\label{thm:common-marginal-averaging}
\label{thm:marginal-random-decomposability}
An index $I:\Mn\to\R_{\geq0}$ satisfies Scale Invariance, Marginal Monotonicity, Random-Benchmark Monotonicity, Marginal Random Decomposability, and Continuity if and only if there exist constants $\alpha>0$ and $\beta\geq0$ such that
\[
I(M)=\alpha I_L^n(M)+\beta
=\alpha\frac{\sigma(M)}{\rho(M)}+\beta
\qquad
\text{for every }M\in\Mn.
\]
\end{theorem}

\begin{proof}
Write
\[
L(M):=I_L^n(M)=\frac{\sigma(M)}{\rho(M)}.
\]

We first prove the ``if'' direction. Suppose that
\[
I(M)=\alpha L(M)+\beta
\]
for some $\alpha>0$ and $\beta\geq0$.

Scale Invariance follows because both $\sigma$ and $\rho$ are homogeneous of degree one. To verify Random-Benchmark Monotonicity, suppose that
\[
\rho(M)\leq\rho(M')
\qquad\text{and}\qquad
\sigma(M)\geq \sigma(M').
\]
Then
\[
L(M)=\frac{\sigma(M)}{\rho(M)}
\geq
\frac{\sigma(M')}{\rho(M')}
=L(M'),
\]
and hence $I(M)\geq I(M')$.

For Marginal Monotonicity, take $M\in\Mnpos$, distinct $p,q\in T$, and $\varepsilon\in\R$ such that $M^{pq,\varepsilon}\in\Mnpos$. Since the elementary change preserves the row and column marginal vectors,
\[
\rho(M^{pq,\varepsilon})=\rho(M),
\]
whereas
\[
\sigma(M^{pq,\varepsilon})=\sigma(M)+2\varepsilon.
\]
Therefore,
\[
I(M^{pq,\varepsilon})-I(M)
=
\frac{2\alpha\varepsilon}{\rho(M)},
\]
which is positive if and only if $\varepsilon>0$.

To verify Marginal Random Decomposability, let $M,M'\in\Mnpos$ have the same normalized row and column marginal distributions. Set
\[
a:=|M|,
\qquad
b:=|M'|,
\]
and write their common normalized marginal vectors as $(u_i)_{i\in T}$ and $(v_i)_{i\in T}$. If $r:=\sum_i u_iv_i$, then
\[
\rho(M)=ar,
\qquad
\rho(M')=br,
\qquad
\rho(M+M')=(a+b)r
=
\rho(M)+\rho(M').
\]
Using the additivity of $\sigma$, we obtain
\begin{align*}
\frac{\rho(M)}{\rho(M+M')}I(M)
+
\frac{\rho(M')}{\rho(M+M')}I(M')
&=
\frac{\alpha \sigma(M)+\beta\rho(M)}{\rho(M+M')}
+
\frac{\alpha \sigma(M')+\beta\rho(M')}{\rho(M+M')}
\\
&=
\alpha\frac{\sigma(M+M')}{\rho(M+M')}+\beta
\\
&=
I(M+M').
\end{align*}
Thus, Marginal Random Decomposability holds. Continuity is immediate because $\sigma$ and $\rho$ are continuous and $\rho$ is strictly positive on $\Mn$.

We next prove the ``only if'' direction. Suppose that $I$ satisfies the five axioms.
As in the proof of Proposition~\ref{prop-monotonic-ALR}, Scale Invariance and Random-Benchmark Monotonicity imply that there exists a nondecreasing function $f$, defined on the range of $L$, such that
\[
I(M)=f(L(M))
\qquad
(M\in\Mn).
\]
For completeness, if $L(M)=L(M')$, scaling $M$ by $\rho(M')/\rho(M)$ makes both its benchmark and its diagonal mass equal to those of $M'$. Random-Benchmark Monotonicity applied in both directions, followed by Scale Invariance, then gives $I(M)=I(M')$. The same normalization argument shows that $f$ is nondecreasing.

Let
\[
\mathcal R_+:=\{L(M):M\in\Mnpos\}.
\]
As shown in the proof of Proposition~\ref{prop-monotonic-ALR}, $\mathcal R_+$ is an open interval. We show that $f$ is locally affine on this interval. Fix $x_0\in\mathcal R_+$ and choose $M^0\in\Mnpos$ with $L(M^0)=x_0$. By Scale Invariance, we may assume that
\[
\rho(M^0)=1.
\]
Fix distinct $p,q\in T$. For every $s$ in a sufficiently small open interval $U$ containing $x_0$, define
\[
M_s:=(M^0)^{pq,(s-x_0)/2}.
\]
The interval $U$ can be chosen so that $M_s\in\Mnpos$ for every $s\in U$. All matrices in this family have the same row and column marginal vectors, and
\[
\rho(M_s)=1,
\qquad
L(M_s)=s.
\]

Take arbitrary $x,y\in U$. Since $M_x$ and $M_y$ are positive and have the same marginal vectors, Marginal Random Decomposability gives
\[
I(M_x+M_y)
=
\frac{I(M_x)+I(M_y)}{2},
\]
because
\[
\rho(M_x)=\rho(M_y)=1
\qquad\text{and}\qquad
\rho(M_x+M_y)=2.
\]
By Scale Invariance and the definition of the family,
\[
I(M_x+M_y)
=
I\left(\frac{M_x+M_y}{2}\right)
=
I\left(M_{(x+y)/2}\right).
\]
Therefore,
\begin{equation}
\label{eq:mrd-midpoint}
f\left(\frac{x+y}{2}\right)
=
\frac{f(x)+f(y)}{2}
\qquad
(x,y\in U).
\end{equation}
Since $f$ is nondecreasing, the midpoint equation implies that $f$ is affine on $U$; see, for example, \citet[\text{ }Chapter~1, Exercise~1, p.~21]{sahoo2011introduction}. Thus, $f$ is locally affine throughout $\mathcal R_+$. Affine representations on overlapping open intervals have the same slope and intercept, and $\mathcal R_+$ is connected. Hence, there exist constants $\alpha$ and $\beta$ such that
\[
f(x)=\alpha x+\beta
\qquad
(x\in\mathcal R_+).
\]
Marginal Monotonicity implies $\alpha>0$. Consequently,
\[
I(M)=\alpha\frac{\sigma(M)}{\rho(M)}+\beta
\qquad
(M\in\Mnpos).
\]

It remains to extend the representation to the boundary. Fix $M\in\Mn$ and define
\[
M_k:=M+\frac{1}{k}\one
\qquad
(k=1,2,\ldots).
\]
Then $M_k\in\Mnpos$ and $M_k\to M$. By Continuity of $I$ and of $\sigma/\rho$ on $\Mn$,
\[
I(M)
=
\lim_{k\to\infty}I(M_k)
=
\alpha\frac{\sigma(M)}{\rho(M)}+\beta.
\]
Finally, applying this representation to the matrix $H^0$ defined in the proof of Theorem~\ref{thm:alr-without-rd} gives
\[
\beta=I(H^0)\geq0.
\]
This completes the proof.
\end{proof}

\subsection{Axiomatization of weighted ALR}\label{subsec:weighted-alr}

We now turn to the second contribution: an axiomatization of the similarity-weighted ALR introduced by
\citet[Section~II, equation~(9)]{choi2025sorting}.
Their approach assigns a similarity score to each pair of types and compares
observed total similarity with its random-matching benchmark. Our contribution
here is an axiomatic characterization of this measure.
The weighted version is useful when types have an ordinal or metric interpretation. For example, treating each age as a separate type yields a type space with many closely related categories. Under the unweighted ALR, a match between ages 30 and 31 and a match between ages 30 and 40 are both treated simply as off-diagonal, different-type matches; neither contributes to the observed same-type numerator. This is restrictive: the former match is much closer in the age space and should naturally contribute more to assortativeness than the latter. Weighted ALRs allow us to encode such similarity while preserving the comparison between observed matching and the random-matching benchmark.

Following \citet{choi2025sorting}, fix
\[
W=(w_{ij})_{i,j=1}^n\in \R_{\geq 0}^{n\times n}\setminus\{0\}.
\]
We call $W$ a \emph{similarity matrix}, or weight matrix. Its entry $w_{ij}$ specifies how much similarity is attributed to a match between a type-$i$ man and a type-$j$ woman: each such match contributes $w_{ij}$ units to measured similarity, so a larger value represents a more similar type pair, while $w_{ij}=0$ means that the pair makes no contribution.
For example, for unordered types one may use $w_{ij}=1\{i=j\}$, whereas for ordered types one may use $w_{ij}=\phi(|i-j|)$ with a nonnegative weakly decreasing function $\phi$. Only relative weights matter, because multiplying every entry of $W$ by the same positive constant leaves weighted ALR and the axioms below unchanged. Thus, our use of arbitrary nonnegative weights is equivalent, up to rescaling, to Choi and Kim's normalization of similarity scores to $[0,1]$.

Define the observed similarity mass by
\[
\sigma_W(M):=\sum_{i=1}^n\sum_{j=1}^n w_{ij}m_{ij}.
\]
We define the \textbf{weighted random-matching benchmark} by
\[
\rho_W(M)=\sum_{i=1}^n\sum_{j=1}^n w_{ij}\frac{x_i(M)y_j(M)}{|M|}.
\]
It is the expected total similarity mass under uniform random matching with the observed row and column marginals held fixed.

Let
\[
\Mn_W:=\left\{M\in \R_{\geq 0}^{n\times n}\setminus\{0\}:\rho_W(M)>0\right\}.
\]
Because $W$ is nonnegative and nonzero, every positive matching matrix belongs to $\Mn_W$; hence $\Mnpos\subseteq\Mn_W$.

In our notation, the \textbf{weighted aggregate likelihood ratio} (weighted ALR)
of \citet[equation~(9)]{choi2025sorting} is
\begin{equation}\label{eq:weighted-alr}
I^n_{L,W}(M):=\frac{\sigma_W(M)}{\rho_W(M)}
=\frac{\sum_{i,j}w_{ij}m_{ij}}
{\sum_{i,j}w_{ij}x_i(M)y_j(M)/|M|}
\qquad(M\in\Mn_W).
\end{equation}
Like the ordinary ALR, this index compares observed similarity with the similarity predicted by random matching with the same row and column marginals. The value $I^n_{L,W}(M)=1$ means that observed weighted similarity is exactly equal to its random benchmark; values above one indicate more weighted similarity than random matching would predict, while values below one indicate less.

The random-benchmark approach developed in Section~\ref{subsec:alr-without-rd} extends naturally to this setting: replace the observed same-type mass and its benchmark by $\sigma_W$ and $\rho_W$, respectively. Throughout this subsection, an index is a function $I:\Mn_W\to\R_{\geq0}$.

\begin{itemize}
\item \textbf{Weighted Random-Benchmark Monotonicity.}
For every $M,M'\in\Mn_W$, if
\[
\rho_W(M)\leq\rho_W(M')
\qquad\text{and}\qquad
\sigma_W(M)\geq\sigma_W(M'),
\]
then
\[
I(M)\geq I(M').
\]
\end{itemize}
The axiom requires the index to rise weakly when observed similarity increases weakly while the similarity predicted by random matching decreases weakly.

For distinct types $p,q\in T$, define
\[
\Delta_W(p,q):=w_{pp}+w_{qq}-w_{pq}-w_{qp}.
\]
The change $M\mapsto M^{pq,\eps}$ leaves $\rho_W$ fixed and changes $\sigma_W$ by $\eps\Delta_W(p,q)$. We assume that
\begin{equation}
\label{eq:weighted-pairwise-assortativity}
\Delta_W(p,q)\geq0\quad\text{for all }p\ne q,
\qquad
\Delta_W(p_0,q_0)>0\quad\text{for some }p_0\ne q_0.
\end{equation}
We call \eqref{eq:weighted-pairwise-assortativity} the \textbf{weak pairwise assortativity condition}. Replacing one type-$p$--type-$q$ match and one type-$q$--type-$p$ match by one type-$p$--type-$p$ match and one type-$q$--type-$q$ match weakly increases measured similarity for every pair of distinct types and strictly increases it for at least one such pair.

\begin{itemize}
\item \textbf{Marginal Similarity Monotonicity.}
For every $M\in\Mnpos$, every pair of distinct types $p,q\in T$,
and every $\eps\in\R$ such that $M^{pq,\eps}\in\Mnpos$,
\[
\operatorname{sgn}\bigl(I(M^{pq,\eps})-I(M)\bigr)
=\operatorname{sgn}\bigl(\eps\Delta_W(p,q)\bigr).
\]
\end{itemize}
Thus, an elementary change raises the index when it raises measured similarity, lowers it when it lowers measured similarity, and leaves it unchanged when measured similarity is unchanged.

Scale Invariance and Continuity are understood as in Section~\ref{subsec:matching-index-axioms}, and Increment Invariance as in Section~\ref{subsec:alr-without-rd}, with the domain $\Mn$ replaced by $\Mn_W$. In particular, Increment Invariance compares the same elementary change in two positive matching matrices with identical row and column marginal vectors.
\begin{theorem}
\label{thm:weighted-alr}
Suppose that $W$ satisfies the weak pairwise assortativity condition
\eqref{eq:weighted-pairwise-assortativity}. An index
$I:\Mn_W\to\R_{\geq0}$ satisfies Scale Invariance,
Weighted Random-Benchmark Monotonicity, Marginal Similarity Monotonicity,
Increment Invariance, and Continuity if and only if there exist
constants $\alpha>0$ and $\beta\in\R$ such that
\[
I(M)=\alpha I^n_{L,W}(M)+\beta
=\alpha\frac{\sigma_W(M)}{\rho_W(M)}+\beta
\qquad
\text{for every }M\in\Mn_W.
\]
\end{theorem}
The constants must make $\alpha I^n_{L,W}(M)+\beta$ nonnegative
on the stated domain. The theorem does not otherwise require $\beta\geq0$.

\begin{proof}
Write
\[
L_W(M):=I^n_{L,W}(M)
=\frac{\sigma_W(M)}{\rho_W(M)}.
\]

For sufficiency, suppose that
$I(M)=\alpha L_W(M)+\beta$ with $\alpha>0$ and $\beta\in\R$.
Scale Invariance follows because $\sigma_W$ and $\rho_W$
are homogeneous of degree one. Continuity follows because
both functions are continuous and $\rho_W>0$ on $\Mn_W$.
Weighted Random-Benchmark Monotonicity holds because
increasing a nonnegative numerator and decreasing a positive
denominator weakly increases their ratio.

For every feasible elementary change,
\[
I(M^{pq,\eps})-I(M)
=
\frac{\alpha\eps\Delta_W(p,q)}{\rho_W(M)}.
\]
Since $\alpha>0$ and $\rho_W(M)>0$, this difference has the
same sign as $\eps\Delta_W(p,q)$, proving Marginal Similarity Monotonicity.
If $M$ and $M'$ have the same row and column marginal vectors,
then $\rho_W(M)=\rho_W(M')$. The displayed formula therefore
also proves Increment Invariance.

Conversely, suppose that the five axioms hold.

Scale Invariance and Weighted Random-Benchmark Monotonicity imply that
\[
I(M)=f(L_W(M))
\qquad(M\in\Mn_W)
\]
for a nondecreasing function $f$ on the range of $L_W$.
Indeed, normalizing a matrix by its benchmark gives
\[
\rho_W\left(\frac{M}{\rho_W(M)}\right)=1,
\qquad
\sigma_W\left(\frac{M}{\rho_W(M)}\right)=L_W(M).
\]
Thus, if two matrices have the same ratio, their normalized
matrices have identical observed similarity and benchmark.
Weighted Random-Benchmark Monotonicity applied in both
directions, together with Scale Invariance, gives identical
index values. The same normalization shows that $f$ is
nondecreasing.

Let
\[
\mathcal R_{W,+}
:=
\{L_W(M):M\in\Mnpos\}.
\]
This set is an interval because $\Mnpos$ is connected and
$L_W$ is continuous. By \eqref{eq:weighted-pairwise-assortativity},
fix distinct $p,q\in T$ with $\Delta_W(p,q)>0$.
The interval is open because, at every positive matrix, a
sufficiently small elementary change of either sign for this pair is feasible and
\[
L_W(M^{pq,\eps})
=
L_W(M)
+\frac{\eps\Delta_W(p,q)}{\rho_W(M)}.
\]

Fix $s_0\in\mathcal R_{W,+}$ and choose $M^0\in\Mnpos$
with $L_W(M^0)=s_0$. Scale $M^0$ so that
$\rho_W(M^0)=1$.
On a sufficiently small open interval $U$ containing $s_0$,
define
\[
M_s
:=
(M^0)^{pq,(s-s_0)/\Delta_W(p,q)}.
\]
All matrices $M_s$ are positive and have the same marginal
vectors, and
\[
\rho_W(M_s)=1,
\qquad
L_W(M_s)=s.
\]
Continuity of $I$ makes $f(s)=I(M_s)$ continuous on $U$.

For $x<y$ in $U$, set
\[
m:=\frac{x+y}{2},
\qquad
\eps:=\frac{y-x}{2\Delta_W(p,q)}>0.
\]
Then
\[
M_x^{pq,\eps}=M_m,
\qquad
M_m^{pq,\eps}=M_y.
\]
Increment Invariance gives
\[
f(m)-f(x)=f(y)-f(m),
\]
or, equivalently,
\[
f\left(\frac{x+y}{2}\right)
=
\frac{f(x)+f(y)}{2}.
\]
Iterating this identity over dyadic weights and using
continuity shows that $f$ is affine on $U$.
Since $\Delta_W(p,q)>0$ for the chosen pair, Marginal Similarity
Monotonicity gives $f(m)>f(x)$, so this slope is strictly positive.

Since $\mathcal R_{W,+}$ is an interval, these local affine
representations agree on overlaps. Consequently, there exist
constants $\alpha>0$ and $\beta\in\R$ such that
\[
I(M)=\alpha L_W(M)+\beta
\qquad(M\in\Mnpos).
\]

Finally, fix any $M\in\Mn_W$ and define
\[
M_k:=M+k^{-1}\one
\qquad(k=1,2,\ldots).
\]
Then $M_k\in\Mnpos$ and $M_k\to M$.
Continuity of $I$ and $L_W$ therefore gives
\[
I(M)
=
\lim_{k\to\infty}I(M_k)
=
\alpha L_W(M)+\beta.
\]
This proves the representation on all of $\Mn_W$.
\end{proof}

We consider two special cases of $I^n_{L,W}$. If
\[
w_{ij}=\mathbf{1}\{i=j\},
\]
then
\[
\sigma_W(M)=\sum_i m_{ii},
\qquad
\rho_W(M)=\sum_i\frac{x_i(M)y_i(M)}{|M|},
\]
and $I^n_{L,W}$ becomes the ordinary multi-type aggregate likelihood ratio. The weak pairwise assortativity condition holds because
\[
w_{pp}+w_{qq}-w_{pq}-w_{qp}=2
\qquad
(p\neq q).
\]

If
\[
w_{ij}=\phi(|i-j|)
\]
for a nonnegative, weakly decreasing function $\phi$ satisfying
\[
\phi(0)>\phi(d)
\]
for at least one realized distance $d>0$, then
\[
I^n_{L,\phi}(M)
=
\frac{\sum_{i,j}\phi(|i-j|)m_{ij}}
{\sum_{i,j}\phi(|i-j|)x_i(M)y_j(M)/|M|}.
\]
This measures whether observed matches are closer in type space than one would expect under random matching with the same marginals. The weak pairwise assortativity condition holds because
\[
\Delta_W(p,q)
=
2\bigl(\phi(0)-\phi(|p-q|)\bigr)\geq0
\]
for every $p\ne q$, with strict inequality for at least one pair.

\begin{remark}
In this section, we have extended the approach in Section~\ref{subsec:alr-without-rd}
rather than that in Section~\ref{subsec:axiomatization-alr}, to weighted ALRs. 
This is because the latter, based on Random Decomposability, does not
yield a straightforward axiomatization of weighted ALRs.
To extend Theorem~\ref{thm:multitype-alr} to general weighted
ALRs, we cannot retain Type Invariance and Full Heterogamy:
with a fixed similarity matrix, relabeling types can change
measured assortativeness, and a matching composed entirely of
different-type pairs need not be minimally assortative.
\end{remark}

\section{ALR in many-to-many matching markets}
\label{subsec:many-to-many-alr}
This section presents our third theoretical contribution: an extension
of ALR to many-to-many matching markets. In the preceding sections,
the random-matching benchmark of ALR---denoted by $\rho(M)$ in the
unweighted case and $\rho_W(M)$ in the weighted case---is derived
for one-to-one markets. For many-to-many matching markets, we instead
construct the benchmark using \emph{simple-graph random matching}.\footnote{A matching can be
represented as a bipartite graph, with men and women as vertices and
partnerships as links. Such a graph is called \emph{simple} if each
man--woman pair is connected by at most one link.}
An individual's degree is the number of their partners, which is
observable in our dataset. Under this procedure, we draw uniformly
from all matchings that preserve each individual's observed degree
and contain at most one link per man--woman pair. The benchmark is
then the expected number of same-type links, or the expected total
similarity in the weighted case. A key challenge is that we are not
aware of a tractable closed-form expression for this expectation.
We therefore introduce a \emph{configuration-matching benchmark},
which admits an explicit formula, and use it to approximate the
simple-graph benchmark.

Figure~\ref{fig:benchmark_compare} illustrates the difference between the two random-matching procedures. Figure~\ref{fig:benchmark_compare}(a) shows one possible matching of four men and four women with prescribed individual degrees. To construct configuration matching, we create a number of labeled copies equal to each individual's degree and draw uniformly from all perfect one-to-one matchings between the male and female copies. Figure~\ref{fig:benchmark_compare}(b) shows one realization of this procedure for the same individuals and degrees. The key difference is that configuration matching allows repeated links between the same original man and woman. Specifically, in Figure~\ref{fig:benchmark_compare}(b), the two copies of the degree-2 man are matched to the two copies of the degree-2 woman, producing two links between these individuals. Such an outcome is excluded under the procedure illustrated in Figure~\ref{fig:benchmark_compare}(a). Allowing repeated links changes the random-matching distribution and can therefore change the benchmark. We derive finite-market error bounds and show that, under suitable conditions, the difference between the two benchmarks becomes negligible relative to the number of links as the market grows.

\begin{figure}[!htbp]
  \begin{center}
  \includegraphics[width = 0.8\textwidth]{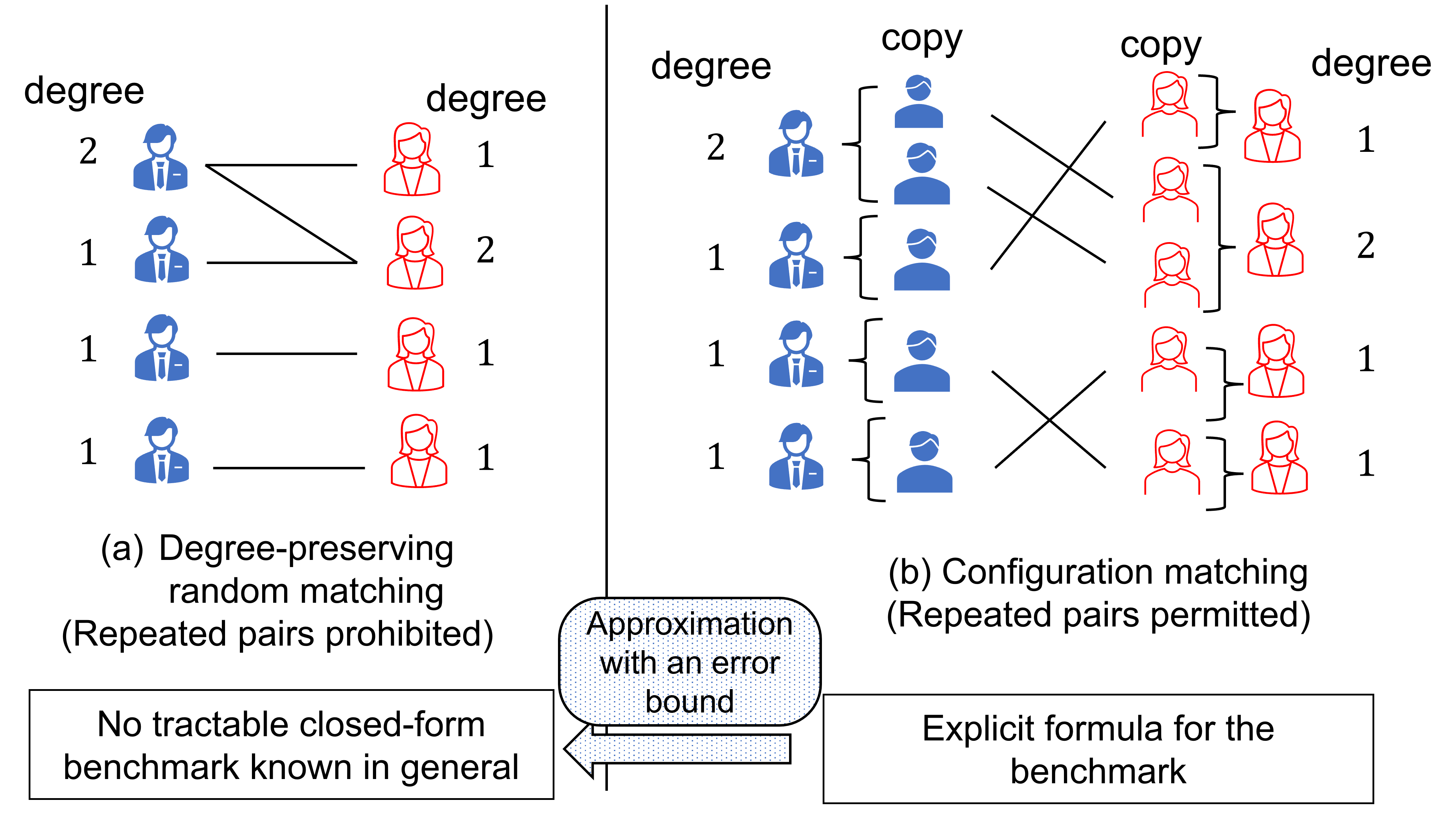}
  \caption{Simple-graph random matching and configuration random matching. Repeated links between the same man and woman are prohibited in (a) and permitted in (b).}
  \label{fig:benchmark_compare}
  \end{center}
\end{figure}

This section proceeds as follows. Section~\ref{subsec:many-dp-benchmark} defines simple-graph random matching and many-to-many ALR. Section~\ref{subsec:many-cm-benchmark} introduces configuration matching, and Section~\ref{subsec:many-finite-bound} bounds the difference between the two benchmarks. Finally, Section~\ref{subsec:many-cm-axiomatization} provides an axiomatization of the configuration-based many-to-many ALR.

\subsection{Simple-graph random matching and the many-to-many ALR}
\label{subsec:many-dp-benchmark}

Let $U$ and $V$ be finite sets of men and women, with types
$\tau(u),\vartheta(v)\in T=\{1,\ldots,n\}$. For the simple-graph benchmark, an observed matching is a simple
bipartite graph $G=(U,V,E)$: a given man--woman pair can appear at most once,
but an individual can belong to several pairs. Write
\[
d_u:=\#\{v:(u,v)\in E\},\qquad
e_v:=\#\{u:(u,v)\in E\},\qquad
L:=|E|=\sum_{u\in U}d_u=\sum_{v\in V}e_v>0.
\]
Individuals with degree zero may be omitted. We assume that all man--woman
pairs are eligible and hold the individuals, their types, and their degrees
fixed when constructing the benchmark.

We define \emph{simple-graph random matching} as follows.
Let $\mathcal G(d,e)$ be the nonempty set of all labeled simple bipartite
graphs on $U,V$ with degrees $(d,e)$, and draw $G^{\mathrm S}$ uniformly
from this set. Thus, the procedure preserves every individual's degree
while ruling out repeated links between the same man and woman.

For a nonnegative, nonzero weight matrix $W$, define
\[
\sigma_W(G):=\sum_{(u,v)\in E}w_{\tau(u),\vartheta(v)},
\qquad
\rho^{\mathrm S}_W(G):=
\mathbb E\bigl[\sigma_W(G^{\mathrm S})\bigr].
\]
The \textbf{simple-graph many-to-many ALR} is
\begin{equation}\label{eq:many-dp-alr}
I^{\mathrm{S}}_{L,W}(G):=
\frac{\sigma_W(G)}{\rho^{\mathrm S}_W(G)},
\qquad \rho^{\mathrm S}_W(G)>0.
\end{equation}
For $w_{ij}=\mathbf1\{i=j\}$, this is the observed number of same-type links
divided by its expectation under simple-graph random matching; we omit
$W$ from the notation in this case. More generally, if type $i$ denotes an
age $a_i$, the weights $w_{ij}=\mathbf1\{|a_i-a_j|\le h\}$ count links
between people whose ages differ by at most $h$. Thus, the definition also
covers the age-gap version of ALR used in our application. In either case,
the numerator counts links, so an individual contributes once for each
partner. A value of one indicates agreement with the corresponding random
benchmark.

\subsection{Configuration random matching and an explicit benchmark}
\label{subsec:many-cm-benchmark}

We define \emph{configuration random matching} as follows.
Replace each man $u$ by $d_u$ labeled copies, and each woman $v$ by $e_v$
labeled copies. Each copy represents one endpoint of a link and is often
called a \emph{stub}. Match the $L$ male stubs to the $L$ female stubs
by a uniformly chosen bijection, and then identify copies belonging
to the same individual.
This construction preserves every individual's degree but can produce
several links between the same man and woman. All such links are counted
with their multiplicities; in particular, repeated same-type links each
contribute to the numerator. The distribution is uniform over stub
bijections, not over the resulting multigraphs.

Let $X_{uv}$ be the number of links between $u$ and $v$ under this model.
Each of $u$'s $d_u$ stubs is paired with a stub of $v$ with probability
$e_v/L$, so
\begin{equation}\label{eq:many-cm-edge-expectation}
\mathbb E_{\mathrm{CM}}X_{uv}=\frac{d_ue_v}{L}.
\end{equation}
Define the type-by-type link table and its marginals by
\[
m_{ij}(G):=\#\{(u,v)\in E:\tau(u)=i,\ \vartheta(v)=j\},
\qquad M(G):=(m_{ij}(G))_{i,j\in T},
\]
\[
x_i:=\sum_{u:\tau(u)=i}d_u=\sum_jm_{ij}(G),
\qquad
y_j:=\sum_{v:\vartheta(v)=j}e_v=\sum_im_{ij}(G).
\]
Here $x_i,y_j$ count link endpoints, rather than individuals, and
$|M(G)|=L$. Equation~\eqref{eq:many-cm-edge-expectation} gives the explicit
\textbf{configuration-matching benchmark}
\begin{equation}\label{eq:many-cm-benchmark}
\rho^{\mathrm{CM}}_W(G)
:=\sum_{u,v}w_{\tau(u),\vartheta(v)}\frac{d_ue_v}{L}
=\frac{1}{L}\sum_{i,j}w_{ij}x_i y_j.
\end{equation}
Consequently, the \textbf{configuration-based many-to-many ALR} is
\begin{equation}\label{eq:many-cm-alr}
I^{\mathrm{CM}}_{L,W}(G)
:=\frac{\sigma_W(G)}{\rho^{\mathrm{CM}}_W(G)}
=\frac{\sum_{i,j}w_{ij}m_{ij}(G)}
       {\sum_{i,j}w_{ij}x_i y_j/L}
=I^n_{L,W}(M(G)),
\qquad \rho^{\mathrm{CM}}_W(G)>0.
\end{equation}
In particular,
\[
\rho^{\mathrm{CM}}(G)=\frac{1}{L}\sum_i x_i y_i,
\qquad
I^{\mathrm{CM}}_L(G)=
\frac{\sum_i m_{ii}(G)}{\sum_i x_i y_i/L}.
\]
The formula has the same algebraic form as in a one-to-one market, but its
marginals reflect individual degrees. The configuration benchmark and index also apply to observed multigraphs: count every link with its multiplicity in the degrees and in $M(G)$. The simple-graph approximation bounds below require their stated simple-graph assumptions. The two random models coincide when
all active individuals have degree one. More generally, conditioning
configuration random matching on having no repeated individual pair gives
the uniform distribution on $\mathcal G(d,e)$: each simple graph has exactly
$\prod_u d_u!\prod_v e_v!$ stub realizations. Because the unconditional model also generates graphs with repeated links, the expectation
in~\eqref{eq:many-cm-benchmark} need not equal the simple-graph
benchmark in~\eqref{eq:many-dp-alr}. 

\subsection{A finite-market approximation bound}\label{subsec:many-finite-bound}
We next bound this difference using only the degree distributions and the
type composition within each degree group. Write
\[
r^{\mathrm S}_W:=\frac{\rho^{\mathrm S}_W(G)}{L},
\qquad
r^{\mathrm{CM}}_W:=\frac{\rho^{\mathrm{CM}}_W(G)}{L},
\qquad
p_i:=\frac{x_i}{L},\quad q_j:=\frac{y_j}{L}.
\]
For identity weights, these are expected same-type shares; for general
weights, they are expected similarity per link.

Let $D_U(t)$ and $D_V(t)$ be the sums of the largest $t$ degrees on the two
sides, with $D_U(0)=D_V(0)=0$ and linear interpolation between integers.
For any positive degree values $k,\ell$ occurring in the market, define
\begin{align*}
\lambda_{k\ell}&:=\frac{k\ell}{L},
&\eta_{k\ell}&:=\frac{D_V(k)+D_U(\ell)}{L},\\
\zeta_{k\ell}&:=
\frac{D_U(D_V(k)/k)+D_V(D_U(\ell)/\ell)}{L},
&t_{k\ell}&:=(1-\zeta_{k\ell})_+,
\end{align*}
where $(z)_+:=\max\{z,0\}$. By symmetry within degree groups, the probability
that a degree-$k$ man and a degree-$\ell$ woman are linked under the uniform
simple-graph model depends only on $k,\ell$; denote it by $\pi_{k\ell}$.
Two-edge and three-edge switching counts, proved in Appendix~\ref{app:edge-bounds}, give
\begin{equation}\label{eq:many-edge-bounds}
\underline\pi_{k\ell}:=
\frac{\lambda_{k\ell}t_{k\ell}}{1+\lambda_{k\ell}t_{k\ell}}
\ \le\ \pi_{k\ell}\ \le\
\frac{\lambda_{k\ell}}
     {\lambda_{k\ell}+(1-\eta_{k\ell})_+}
=: \overline\pi_{k\ell}.
\end{equation}
Let $c_{k\ell}:=(\underline\pi_{k\ell}+\overline\pi_{k\ell})/2$ and
$h_{k\ell}:=(\overline\pi_{k\ell}-\underline\pi_{k\ell})/2$.

Let $n_k$ and $m_\ell$ be the numbers of degree-$k$ men and degree-$\ell$
women, respectively. Let $a_i(k)$ and $b_j(\ell)$ be the proportions of
type $i$ and type $j$ individuals within these groups. Define
\begin{equation}\label{eq:many-centered-sensitivity}
C^W_{k\ell}:=\sum_{i,j}w_{ij}\bigl(a_i(k)-p_i\bigr)
                                      \bigl(b_j(\ell)-q_j\bigr),
\end{equation}
and set
\begin{align}
A_W&:=\frac{1}{L}\sum_{k,\ell}n_km_\ell C^W_{k\ell}
                         (c_{k\ell}-\lambda_{k\ell}),\nonumber\\
R_W&:=\frac{1}{L}\sum_{k,\ell}n_km_\ell |C^W_{k\ell}|h_{k\ell},
\qquad B_W:=|A_W|+R_W.
\label{eq:many-bound-definition}
\end{align}
For identity weights, $C^W_{k\ell}=\sum_i(a_i(k)-p_i)(b_i(\ell)-q_i)$.

\begin{proposition}
\label{prop:many-finite-bound}
For any feasible positive degree sequence and any nonnegative weight matrix
$W$, with all man--woman pairs eligible,
\begin{equation}\label{eq:many-finite-bound}
\bigl|r^{\mathrm S}_W-r^{\mathrm{CM}}_W-A_W\bigr|\le R_W,
\qquad
\bigl|r^{\mathrm S}_W-r^{\mathrm{CM}}_W\bigr|\le B_W.
\end{equation}
Equivalently, $|\rho^{\mathrm S}_W(G)-\rho^{\mathrm{CM}}_W(G)|\le LB_W$.
If type proportions do not vary with degree on at least one side of the market,
then the two benchmarks agree exactly.
\end{proposition}

\begin{proof}
To see the role of the centered type proportions, degree preservation gives
\[
\sum_\ell m_\ell(\pi_{k\ell}-\lambda_{k\ell})=0,
\qquad
\sum_k n_k(\pi_{k\ell}-\lambda_{k\ell})=0.
\]
It follows that
\[
r^{\mathrm S}_W-r^{\mathrm{CM}}_W
=\frac{1}{L}\sum_{k,\ell}n_km_\ell C^W_{k\ell}
                                    (\pi_{k\ell}-\lambda_{k\ell}).
\]
Subtracting $A_W$ and using
$|\pi_{k\ell}-c_{k\ell}|\le h_{k\ell}$ proves the bound.
The second inequality follows by the triangle inequality. If type
proportions are constant across degree groups on at least one side, the
corresponding centered proportions vanish, so $C^W_{k\ell}=0$ for every
$k,\ell$ and the two benchmarks agree.
\end{proof}
All its inputs are available from degree-by-type counts; simulation of
random networks is unnecessary for computing $B_W$.

A simple sufficient condition for a small bound is sparse individual
degrees. In particular, if
\[
\delta:=\frac{d_{\max}e_{\max}}{L}<\frac12,
\qquad
\mu_k:=\frac{kn_k}{L},\quad \nu_\ell:=\frac{\ell m_\ell}{L},
\]
then the bounds in~\eqref{eq:many-edge-bounds} imply
\begin{equation}\label{eq:many-max-degree-bound}
B_W\le\frac{3\delta}{1-2\delta}
             \sum_{k,\ell}\mu_k\nu_\ell|C^W_{k\ell}|.
\end{equation}
The directly computed bound can nevertheless be small even when this
maximum-degree condition fails.

\begin{proposition}
\label{prop:many-asymptotic-bound}
Consider a sequence of markets with $N_U$ active men, $N_V$ active women,
and $L$ links. Suppose that
\[
\frac{L}{N_UN_V}\longrightarrow0,
\qquad
\sup\frac{N_U\sum_ud_u^2}{L^2}<\infty,
\qquad
\sup\frac{N_V\sum_ve_v^2}{L^2}<\infty,
\]
where the suprema are over the market sequence. For a fixed nonnegative
weight matrix $W$, $B_W\to0$ and
$|r^{\mathrm S}_W-r^{\mathrm{CM}}_W|\to0$, uniformly over type assignments.
\end{proposition}

\begin{proof}
Put $\varepsilon:=L/(N_UN_V)$, $a_k:=kN_U/L$, and
$b_\ell:=\ell N_V/L$, and choose $C\ge1$ bounding both normalized
second moments in the assumptions. Cauchy--Schwarz gives
$D_U(t)\le\sqrt{t\sum_ud_u^2}$ and its female counterpart,
also for interpolated arguments. Hence
\[
\eta_{k\ell}\le\sqrt{C\varepsilon}
 (\sqrt{a_k}+\sqrt{b_\ell}),\qquad
\zeta_{k\ell}\le C^{3/4}\varepsilon^{1/4}
 (a_k^{-1/4}+b_\ell^{-1/4}),\qquad
\lambda_{k\ell}=\varepsilon a_kb_\ell.
\]
Let $\mathbb E_{\mu\nu}$ denote expectation under the product
distribution $\mu_k\nu_\ell$. Since $\mathbb E_\mu a_k\le C$ and
\[
\mathbb E_\mu a_k^{-1/4}
=\sum_k\frac{n_k}{N_U}a_k^{3/4}\le1,
\]
with analogous bounds for $b_\ell$, Jensen's inequality yields
$\mathbb E_{\mu\nu}(\zeta_{k\ell}+\lambda_{k\ell}
+2\eta_{k\ell})=O(\varepsilon^{1/4})$.

The edge bounds in~\eqref{eq:many-edge-bounds} imply
\[
|c_{k\ell}-\lambda_{k\ell}|+h_{k\ell}
\le\lambda_{k\ell}
 (\zeta_{k\ell}+\lambda_{k\ell}+2\eta_{k\ell})
+\mathbf1\{\eta_{k\ell}>1/2\}.
\]
Indeed, the lower-bound deficit is at most
$\lambda_{k\ell}(\zeta_{k\ell}+\lambda_{k\ell})$, while the
upper-bound excess is at most $2\lambda_{k\ell}\eta_{k\ell}$
when $\eta_{k\ell}\le1/2$, and at most one otherwise.
Moreover, $\eta_{k\ell}>1/2$ requires $k>N_V/(16C)$ or
$\ell>N_U/(16C)$. The second-moment bounds therefore give
\[
\frac1L\sum_{\eta_{k\ell}>1/2}n_km_\ell
\le\frac{256C^2}{L}
 \left(\frac{\sum_ud_u^2}{N_V}
       +\frac{\sum_ve_v^2}{N_U}\right)
\le512C^3\varepsilon.
\]
Finally, $|C^W_{k\ell}|\le4\max_{i,j}w_{ij}$ and
$n_km_\ell\lambda_{k\ell}/L=\mu_k\nu_\ell$, so
\[
B_W\le4\max_{i,j}w_{ij}
\left[\mathbb E_{\mu\nu}
 (\zeta_{k\ell}+\lambda_{k\ell}+2\eta_{k\ell})
+512C^3\varepsilon\right]\longrightarrow0.
\]
All estimates are uniform over type assignments.
Proposition~\ref{prop:many-finite-bound} gives the remaining claim.
\end{proof}
This result requires neither $d_{\max}e_{\max}/L\to0$ nor independence
between degree and type. It follows by applying the one-edge bounds to
bounded normalized degrees and controlling the remaining degree groups
using the second-moment assumptions. It concerns absolute benchmark errors. A configuration benchmark bounded
away from zero is sufficient to obtain convergence of relative errors.

To translate a benchmark bound into an index bound, let
$s_W:=\sigma_W(G)/L$. Whenever $r^{\mathrm{CM}}_W>B_W$,
\begin{equation}\label{eq:many-index-error}
\left|I^{\mathrm S}_{L,W}(G)-I^{\mathrm{CM}}_{L,W}(G)\right|
\le
\frac{s_W B_W}
     {r^{\mathrm{CM}}_W\bigl(r^{\mathrm{CM}}_W-B_W\bigr)}.
\end{equation}
Thus, a small absolute error in the benchmark gives a small error in ALR
when its denominator is sufficiently far from zero. The reported numerical
bounds concern benchmark shares, not ALR values themselves.

\subsection{Axiomatization of the configuration-based measure}\label{subsec:many-cm-axiomatization}
Configuration matching is a one-to-one matching of stubs. We can therefore
apply the preceding axiomatizations to link tables, interpreting ``same-type
pairs'' as same-type links counted with multiplicity and ``random matching''
as uniform matching of stubs. This interpretation gives precisely the
explicit benchmark in~\eqref{eq:many-cm-benchmark}.

As in the one-to-one analysis, the axioms are imposed on a continuous domain
of aggregate tables, allowing arbitrary nonnegative real link masses. Write
\[
\mathcal E_W^n:=\left\{M\in\R_{\geq0}^{n\times n}\setminus\{0\}:
\rho_W^{\mathrm{CM}}(M)>0\right\},
\qquad
\rho_W^{\mathrm{CM}}(M):=
\frac{\sum_{i,j}w_{ij}x_i(M)y_j(M)}{|M|}.
\]
Here $\sigma_W(M)=\sum_{i,j}w_{ij}m_{ij}$ and
$I^{\mathrm{CM}}_{L,W}(M):=\sigma_W(M)/\rho_W^{\mathrm{CM}}(M)$.
The index of an observed graph is obtained by evaluating the table index at
$M(G)$. The axioms concern this full table domain, rather than only the
finite collection of simple graphs with one fixed degree sequence.

To include age-gap weights, which can assign equal similarity to nearby
but distinct types, assume
\begin{equation}\label{eq:many-weak-pairwise-condition}
\Delta_W(p,q)\geq0\quad\text{for all }p\ne q,
\qquad
\Delta_W(p_0,q_0)>0\quad\text{for some }p_0\ne q_0.
\end{equation}
The elementary change $M\mapsto M^{pq,\eps}$ is the same as in
Section~\ref{subsec:alr-without-rd}. At the stub level, its integer
version replaces links between types $p,q$ and $q,p$ by links between types
$p,p$ and $q,q$, preserving the degrees of the four endpoints. It preserves
the benchmark and changes similarity by $\eps\Delta_W(p,q)$.

For an index $J:\mathcal E_W^n\to\R_{\geq0}$, consider the following axioms.
\begin{itemize}
\item \textbf{Link Scale Invariance.}
$J(\lambda M)=J(M)$ for every $M\in\mathcal E_W^n$ and $\lambda>0$.

\item \textbf{Configuration Random-Benchmark Monotonicity.}
For every $M,M'\in\mathcal E_W^n$,
\[
\rho_W^{\mathrm{CM}}(M)\le\rho_W^{\mathrm{CM}}(M'),
\quad \sigma_W(M)\ge\sigma_W(M')
\quad\Longrightarrow\quad J(M)\ge J(M').
\]

\item \textbf{Marginal Similarity Monotonicity.}
For every positive $M$, every $p\ne q$, and every $\eps\in\R$ for which
$M^{pq,\eps}$ is positive,
\[
\operatorname{sgn}\bigl(J(M^{pq,\eps})-J(M)\bigr)
=\operatorname{sgn}\bigl(\eps\Delta_W(p,q)\bigr).
\]
Thus, a degree-preserving change increases the index exactly when it
increases measured similarity, and leaves it unchanged when similarity
is unchanged.

\item \textbf{Link Increment Invariance.}
If positive tables $M,M'$ have identical row and column marginal vectors,
then, for every $p\ne q$ and every $\eps>0$ such that both modified tables
are positive,
\[
J(M^{pq,\eps})-J(M)
=J((M')^{pq,\eps})-J(M').
\]

\item \textbf{Continuity.} $J$ is continuous on $\mathcal E_W^n$.
\end{itemize}

Applying Theorem~\ref{thm:weighted-alr} to the table of links between
types gives the following corollary.
\begin{corollary}
\label{thm:many-cm-alr}
Let $n\geq2$, and let $W$ be nonnegative and nonzero and satisfy
\eqref{eq:many-weak-pairwise-condition}. An index
$J:\mathcal E_W^n\to\R_{\geq0}$ satisfies the five axioms above if and only if
there exist constants $\alpha>0$ and $\beta\in\R$ such that
\begin{equation}\label{eq:many-affine-representation}
J(M)=\alpha I^{\mathrm{CM}}_{L,W}(M)+\beta
=\alpha\frac{\sum_{i,j}w_{ij}m_{ij}}
              {\sum_{i,j}w_{ij}x_i(M)y_j(M)/|M|}+\beta
\quad(M\in\mathcal E_W^n),
\end{equation}
with the right-hand side nonnegative throughout the domain.
For $w_{ij}=\mathbf1\{i=j\}$, this characterizes the unweighted
configuration-based many-to-many ALR up to $\alpha>0$ and $\beta\geq0$.
\end{corollary}
\begin{proof}
For every link table $M$, the configuration benchmark
$\rho_W^{\mathrm{CM}}(M)$ equals the weighted random-matching benchmark
$\rho_W(M)$ of Section~\ref{subsec:weighted-alr}. Hence
$\mathcal E_W^n=\Mn_W$ and $I^{\mathrm{CM}}_{L,W}(M)=I^n_{L,W}(M)$
on this domain. Condition~\eqref{eq:many-weak-pairwise-condition}
and the five axioms above therefore coincide with the condition and
axioms of Theorem~\ref{thm:weighted-alr}, which gives the stated
characterization.

For identity weights, a table $M$ with zeros on the diagonal and
ones elsewhere has positive benchmark and zero observed similarity,
so $J(M)=\beta\geq0$.
\end{proof}

\begin{remark}
In this section, we have axiomatized the configuration-based many-to-many ALR $I^{\mathrm{CM}}_{L,W}$, defined in
Section~\ref{subsec:many-cm-benchmark}, rather than the simple-graph
many-to-many ALR $I^{\mathrm{S}}_{L,W}$, defined in
Section~\ref{subsec:many-dp-benchmark}. Axiomatizing the latter index $I^{\mathrm{S}}_{L,W}$ is not straightforward because it depends not only on the aggregate
$n\times n$ matching table but also on the number of partners each
individual has. An axiomatization would therefore require a domain
that includes this individual-level information. In contrast, the former index $I^{\mathrm{CM}}_{L,W}$ uses a benchmark based on uniform one-to-one matching between copies. Since this benchmark depends only on the row and column sums of the $n\times n$ matching table, the index is determined by this table alone, allowing us to apply the preceding axiomatization without changing the domain.
\end{remark}

\section{Application}\label{sec:application}

We apply our measures to data from IBJ, a Japanese marriage matching platform, because its stage-by-stage records of partner formation fill the gap between the two kinds of data usually used to study who matches with whom.\footnote{Appendix~\ref{app:data-details} describes how the platform is organized and how users join it.} Administrative registers and censuses record marriages \citep{eika2019educational, chiappori2025changes, chiappori2026assortative}: they show the final, one-to-one matching but none of the contacts and relationships that preceded it. Online dating data cover the opposite end of the process \citep{hitsch2010matching}: they show who contacts whom, a many-to-many matching, but rarely whether a contact turns into a relationship or a marriage, and the attributes on which sorting is measured are self-reported and unverified, so users can misstate them \citep{toma2008separating}. IBJ links the two ends for the same pairs and verifies key attributes, such as age and income, against official documents; some agencies also collect health checkup records, but height is not uniformly document-verified.\footnote{The platform also records education, occupation, marital history, and stated preferences over children and the division of household work, as well as stated acceptable partner ages and saved search conditions; \citet{inoue2026marital} and \citet{otani2026happy} describe these variables.} We observe every meeting request and follow the resulting pair through dating to engagement or separation, so a single market yields a many-to-many matching at the early stages and an almost one-to-one matching at the final stage. These records make IBJ a natural application for our measures in two respects: they require a measure of assortativeness that remains well defined when users hold many partners at once, and they allow us to compare assortativeness across stages within one market and ask at which stage sorting emerges.

\subsection{Institutional Background}\label{subsec:institution}

Figure \ref{fig:ibj_matching_process} shows that partner formation on IBJ proceeds through four stages and that moving from one stage to the next requires the consent of both parties. In the \emph{Application} stage, a user browses profiles and sends a meeting request; both men and women can send and receive requests. If the recipient accepts, the agencies arrange a first meeting (\textit{omiai}), after which each side separately tells its agency whether to continue. When both agree, the pair enters the \emph{Pre-relationship} stage (which Figure \ref{fig:ibj_matching_process} labels trial dating), a period of dates that is not exclusive: users may date several partners at once and keep searching. If both decide to go further, the pair moves to the \emph{Serious relationship} stage (committed dating), in which both users must end their other trial relationships and stop receiving requests. The process ends at the \emph{Proposal} stage, when both confirm the intention to marry and leave the platform. A pair can break off at any stage, and users whose serious relationship ends return to search. The stages thus move users step by step from open search toward exclusive commitment.

Platform users are not a random sample of people who marry in Japan. Compared with national statistics, couples formed on IBJ are older and more highly educated, and both partners have higher incomes, which is consistent with a fee-based, document-verified service that attracts people who enter the marriage market relatively late. These differences are mainly compositional. When matching patterns are compared using cell-level likelihood ratios---observed pair shares divided by the products of the corresponding male and female type shares--- the age and education matching patterns of IBJ couples are close to those in Vital Statistics and the National Fertility Survey. For income, differences from the Employment Status Survey may also reflect its post-marital measurement of earnings, including reductions in labor supply among some wives. \citet{inoue2026marital} report the full comparison. We therefore read the patterns below as describing sorting within a selected but sizable market rather than in the Japanese population as a whole.

\begin{figure}[!htbp]
  \begin{center}
  \includegraphics[width = \textwidth]{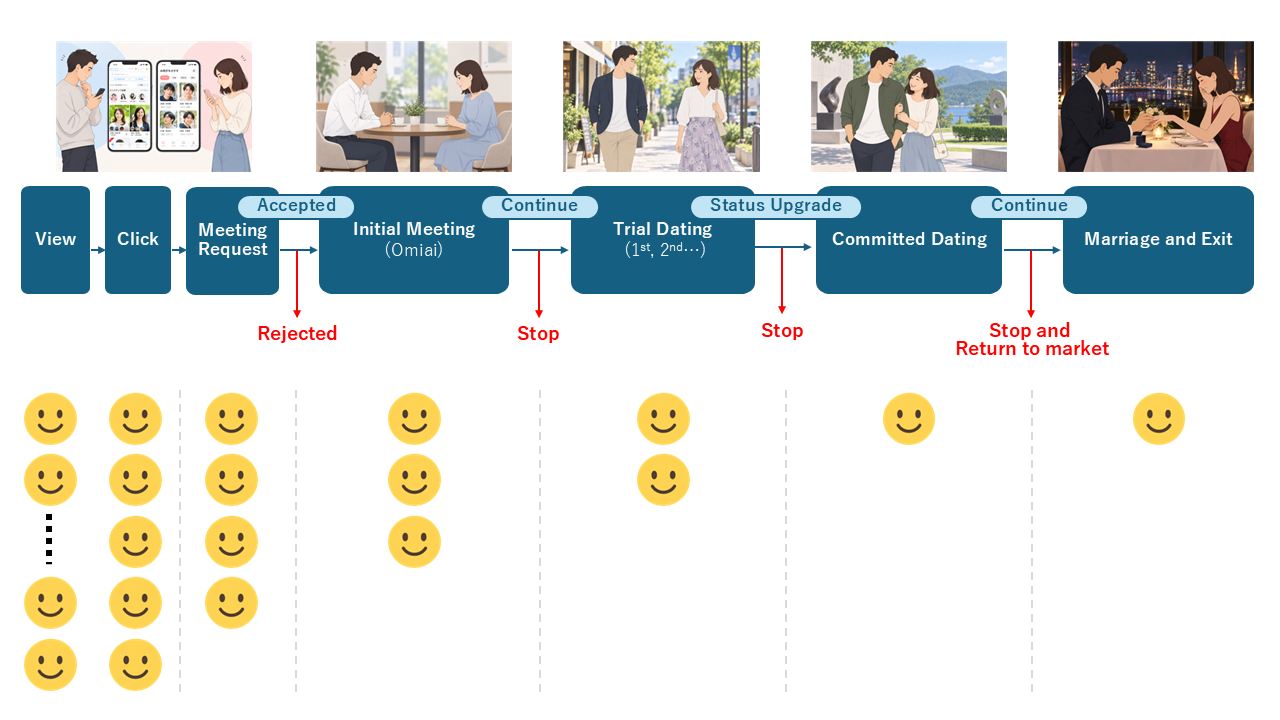}
  \caption{Matching Process}
  \label{fig:ibj_matching_process}
  \end{center}
  \footnotesize
  Note: The four stages used in the analysis correspond to the boxes as follows. Application: a meeting request is sent (Meeting Request), including requests that end before or after the initial meeting (Omiai). Pre-relationship: the pair starts non-exclusive trial dating (Trial Dating). Serious relationship: the pair upgrades to exclusive dating (Committed Dating). Proposal: the pair becomes engaged and leaves the platform (Marriage and Exit). Viewing and clicking on profiles precede the Application stage and are not used.
\end{figure}

At each of the first three stages, we count requests that reach that stage or a later one, grouping them by the year of the initial request. The first meeting is not treated as a separate stage. These stages are many-to-many: users may have several partners over a request cohort, even though serious relationships are exclusive at any point in time. Proposal-stage observations are grouped by engagement year and are almost one-to-one; a small number of member registrations appear in more than one engagement pair. The request cohorts and engagement-year samples therefore do not form a common transition cohort.

\subsection{Data and Descriptive Patterns}\label{subsec:data}

We use administrative records of meeting requests sent on the platform between 2021 and 2024 and of engagements recorded over the same period.\footnote{A user denotes a platform registration. A person who changes agencies or re-registers receives a new registration, and the data do not link these registrations to the same person.} We pool requests sent by men and women and retain repeated requests between the same man and woman. At the first three stages, the year is the year of the initial request, and a request contributes to every stage it reaches in the available records. At the Proposal stage, the year is the engagement year. Age is computed as the relevant calendar year minus birth year; annual income and height are taken from the available member profiles rather than a year-specific profile history. We use one-year age bins from 23 to 59, income categories labeled from 3 to 21 million yen in increments of 1 million yen, and height categories labeled from 142 to 198 cm in increments of 2 cm. The lowest height category also includes shorter users.
A pair whose attribute falls outside these ranges is excluded from the matching matrix for that attribute. These definitions make the matrices described below identical to those on which our assortativeness measures are computed.

\begin{figure}[!htbp]
  \begin{center}
  \subfloat[Age (Share of pairs, \%)]{\includegraphics[width = 0.8\textwidth]{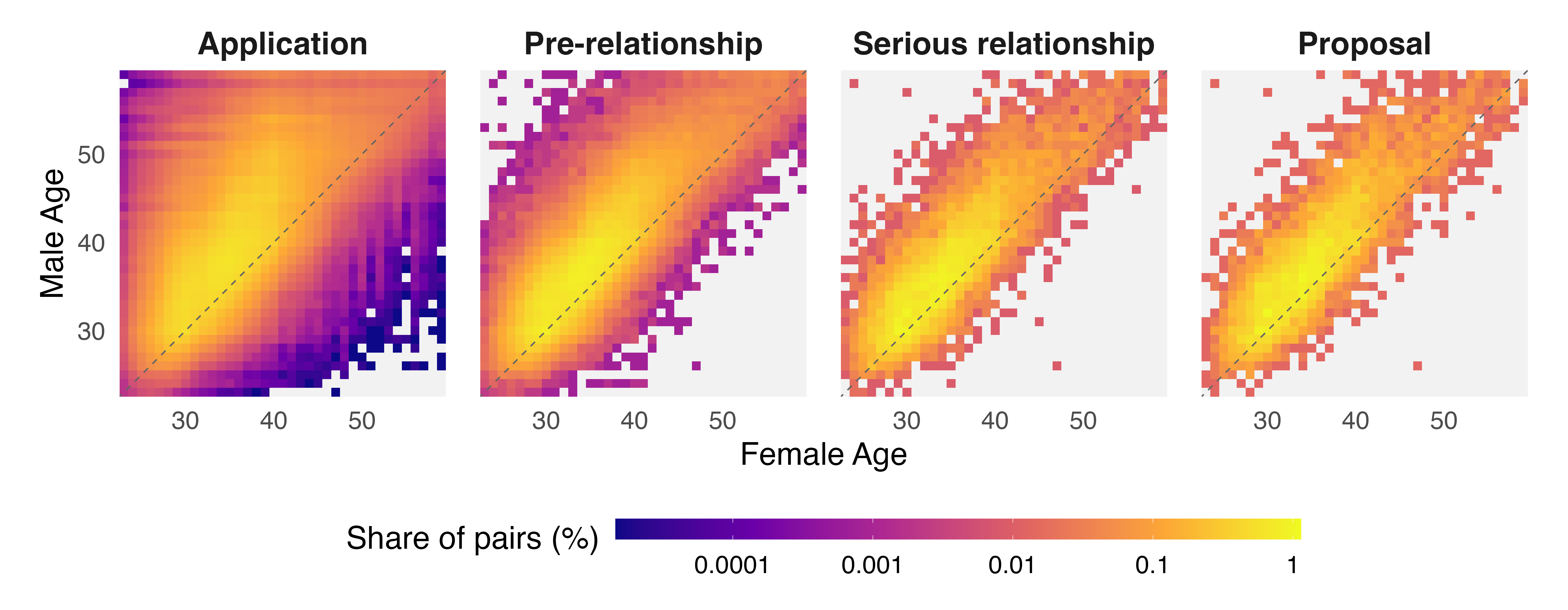}}\\
  \vspace{6mm}
  \subfloat[Income (Share of pairs, \%)]{\includegraphics[width = 0.8\textwidth]{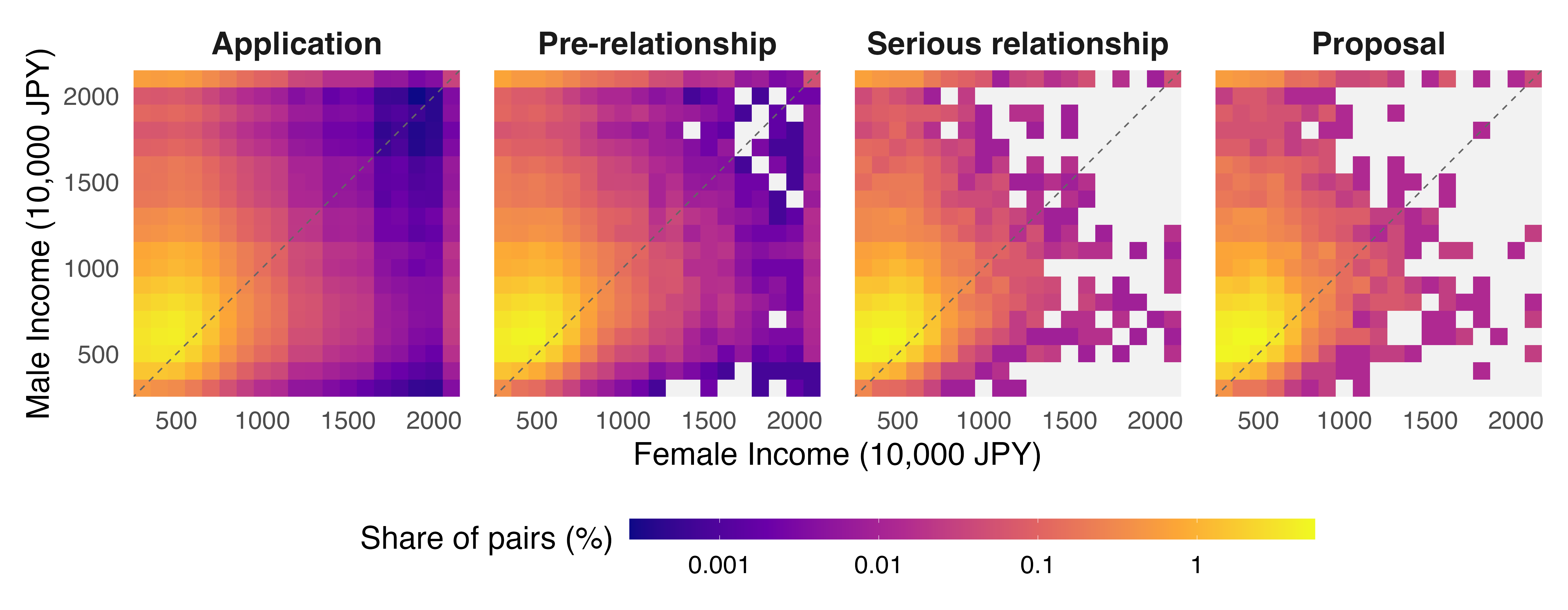}}\\
  \vspace{6mm}
  \subfloat[Height (Share of pairs, \%)]{\includegraphics[width = 0.8\textwidth]{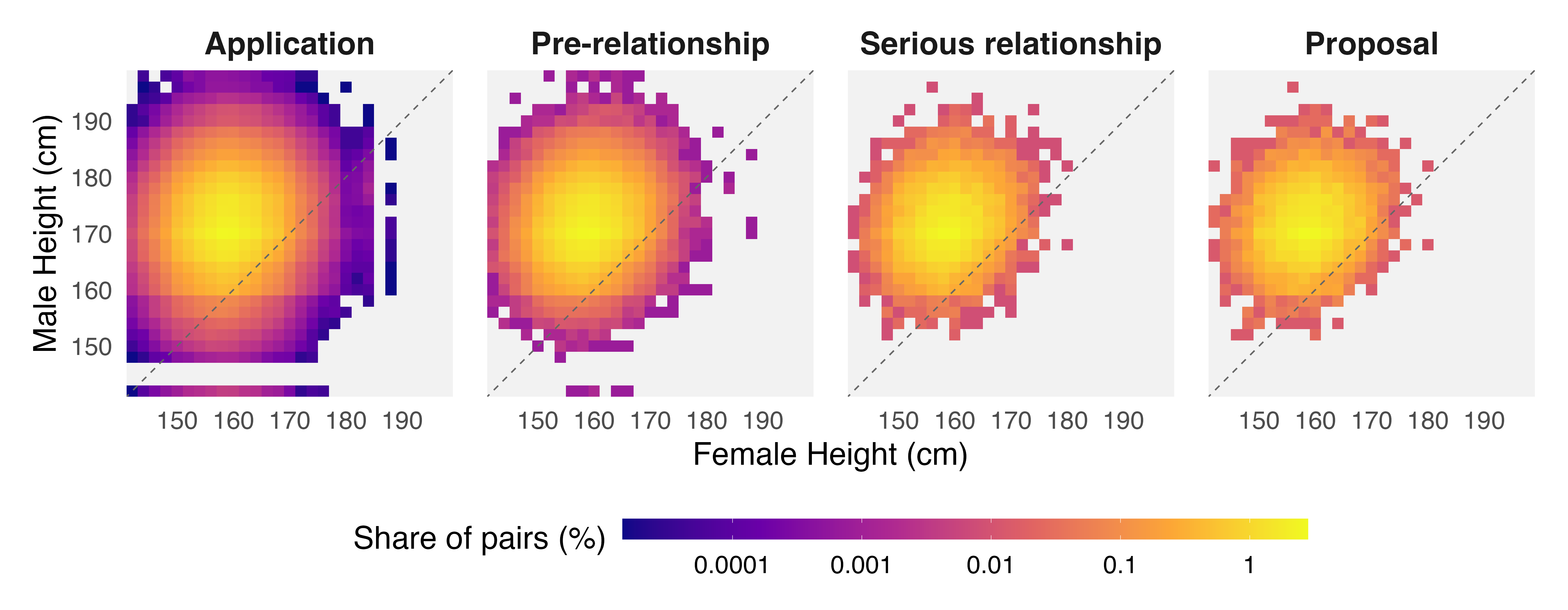}}
  \caption{Matching Matrices by Stage in 2024: Share of Pairs}
  \label{fg:matching_matrix_share_by_stage}
  \end{center}
  \footnotesize
  Note: Each cell shows the share of pairs whose man and woman fall in the given bins, pooling requests sent by men and women. The first three stages use the 2024 request cohort; Proposal uses engagements recorded in 2024. Shares sum to 100 within each stage. Pairs at a stage are all requests that reach that stage or a later one; Proposal-stage pairs are engagements. The color scale is logarithmic and common to the four stages; gray cells contain no pairs, and the dashed line marks cells in which the man and the woman fall in the same bin. Income is in units of 10,000 yen.
\end{figure}

\begin{table}[!htbp]
  \begin{center}
      \caption{Pairs, Users, and Same-Bin Shares by Stage, 2024}
      \label{tb:pair_user_same_bin_share_by_stage}
      
\begin{tabular}{lrrrrrrrr}
\toprule
\multicolumn{4}{c}{ } & \multicolumn{2}{c}{Pairs per user} & \multicolumn{3}{c}{Same-bin share (\%)} \\
\cmidrule(l{3pt}r{3pt}){5-6} \cmidrule(l{3pt}r{3pt}){7-9}
Stage & Pairs & Men & Women & Men & Women & Age & Income & Height\\
\midrule
Application & 6,972,848 & 57,060 & 59,813 & 122.20 & 116.58 & 6.7 & 11.2 & 2.5\\
Pre-relationship & 132,257 & 35,784 & 37,953 & 3.70 & 3.48 & 9.5 & 10.9 & 2.1\\
Serious relationship & 11,929 & 11,252 & 11,300 & 1.06 & 1.06 & 10.6 & 11.1 & 2.0\\
Proposal & 7,577 & 7,576 & 7,576 & 1.00 & 1.00 & 10.4 & 11.3 & 2.2\\
\bottomrule
\end{tabular}

  \end{center}\footnotesize
  \textit{Note}: The first three stages count requests sent in 2024 that reach that stage or a later one, pooling requests sent by men and women and retaining repeated requests. Proposal-stage pairs are engagements recorded in 2024. Men and women count distinct users who appear in at least one pair at the stage, and pairs per user divides the number of pairs by these counts. The same-bin share is the percentage of pairs whose man and woman fall in the same bin, among pairs with both attributes inside the bin range; it equals the sum of the diagonal cells in Figure \ref{fg:matching_matrix_share_by_stage}.
\end{table}

Table \ref{tb:pair_user_same_bin_share_by_stage} shows that the early stages are many-to-many on a large scale and that the later stages are close to one-to-one. In 2024, the Application stage contains 6,972,848 pairs, so that the average man appears in 122.20 pairs and the average woman in 116.58. At the Pre-relationship stage, the averages fall to 3.70 pairs per man and 3.48 per woman, at the Serious relationship stage they are 1.06 for both, and at the Proposal stage they are 1.00. Assortativeness at the early stages is therefore a property of a many-to-many matching, which is why our measures must not presuppose that each user has a single partner. The marginal distributions of the three attributes are broadly stable over 2021--2024; Appendix~\ref{app:marginal-distributions} reports them. %

Figure \ref{fg:matching_matrix_share_by_stage} shows that the stages differ in how strongly pairs concentrate near the diagonal of the matching matrix and that the pattern differs across attributes, and Table \ref{tb:pair_user_same_bin_share_by_stage} reports the share of pairs whose man and woman fall in the same bin. For age, mass is spread widely at the Application stage and becomes concentrated along a band in which the man is the same age as or slightly older than the woman; the same-bin share rises from 6.7 percent at the Application stage to 9.5 percent at the Pre-relationship stage and to 10.6 and 10.4 percent at the Serious relationship and Proposal stages. For income, the same-bin share stays between 10.9 and 11.3 percent at every stage, and the mass remains concentrated among women in low income categories and men in middle income categories. For height, most of the mass lies in the region in which the man is taller than the woman, and the same-bin share is small, falling from 2.5 percent at the Application stage to between 2.0 and 2.2 percent at the later stages. These raw shares, however, depend on the marginal distributions at each stage, which themselves change as pairs are selected out; the next subsection therefore turns to measures that compare observed matching with a random benchmark based on these marginals.

\subsection{Results}\label{subsec:results}

\subsubsection{What we compute}\label{subsubsec:what_we_compute}

For each stage and year, we compute four versions of the
configuration-based ALR defined in Section~\ref{subsec:many-cm-benchmark}.
The first is the unweighted ALR on the fine bins defined in
Section~\ref{subsec:data}. It counts a pair as one if both partners
fall in the same bin and as zero otherwise. The next two use the same
fine bins and the similarity weights $w_{ij}=\mathbf{1}\{|i-j|\leq k\}$
from Section~\ref{subsec:weighted-alr}, with $k=1$ and $k=2$,
respectively. Thus, a pair counts as one whenever the partners' bin
indices differ by at most $k$. The fourth is the unweighted ALR on
coarse bins formed by combining adjacent fine bins into five-year age
groups, four income groups, and ten-centimeter height groups (see the
note to Table~\ref{tb:assortativeness_by_stage}). Our request data are multigraphs because repeated requests are retained. We use the configuration benchmark directly; we do not claim that the reported empirical indices have been validated as approximations to the simple-graph indices.
Following \citet{chiappori2025changes}, we resample engaged couples
with replacement within each year at the Proposal stage, recompute
each index in each bootstrap sample, and adjust $p$-values within
prespecified families using the Romano--Wolf step-down procedure.
The couple bootstrap treats engagement pairs as independent; the few member registrations appearing in multiple engagement pairs make this an approximation.
At the three earlier stages, a user may appear in multiple pairs.
We therefore resample men and women independently. Each observed
pair's weight equals the number of times the man is drawn multiplied
by the number of times the woman is drawn. Because we have not
established the validity of this procedure in our many-to-many
setting, we interpret the resulting $p$-values as suggestive.

\subsubsection{Assortativeness across stages}\label{subsubsec:results_stages}

Table \ref{tb:assortativeness_by_stage} shows that sorting on age strengthens as pairs move from application to engagement, whereas sorting on income and height does not. Panel (a) reports the levels with the four years pooled. For age, the ALR on one-year bins rises from 1.69 at the Application stage to 2.08 at the Pre-relationship stage, 2.29 at the Serious relationship stage, and 2.31 at the Proposal stage, and the weighted and coarse-bin versions follow the same path at lower levels. Panel (b) shows that the two early transitions account for the increase: the differences from Application to Pre-relationship and from Pre-relationship to Serious relationship are positive for all four indices and are distinguishable from zero under the diagnostic bootstrap, whereas the difference from Serious relationship to Proposal is small for every index and is not distinguishable from zero. For income, the indices stay between 1.07 and 1.12 at every stage; the only difference distinguishable from zero is the rise from Application to Pre-relationship, which is at most 0.04. For height, every index lies below one at every stage: Figure \ref{fg:matching_matrix_share_by_stage} shows that pairs concentrate where the man is taller than the woman rather than near the diagonal, so same-bin pairs are rarer than under random matching, and the indices fall from Application to Pre-relationship and are flat thereafter. Age sorting thus emerges when meeting requests turn into first relationships rather than at the final commitment, and income and height show no comparable tightening across stages.

\begin{table}[!htbp]
  \setlength{\tabcolsep}{4pt}
  \begin{center}
      \caption{Assortativeness by Stage, 2021--2024 Pooled}
      \label{tb:assortativeness_by_stage}
      \subfloat[Levels by stage]{\footnotesize
\begin{tabular}{lrrrrrrrrrrrr}
\toprule
\multicolumn{1}{c}{ } & \multicolumn{4}{c}{Age} & \multicolumn{4}{c}{Income} & \multicolumn{4}{c}{Height} \\
\cmidrule(l{3pt}r{3pt}){2-5} \cmidrule(l{3pt}r{3pt}){6-9} \cmidrule(l{3pt}r{3pt}){10-13}
 & App & Pre & Ser & Pro & App & Pre & Ser & Pro & App & Pre & Ser & Pro\\
\midrule
ALR, fine bins & 1.69 & 2.08 & 2.29 & 2.31 & 1.09 & 1.12 & 1.08 & 1.08 & 0.92 & 0.83 & 0.82 & 0.82\\
Weighted ALR, $k = 1$ & 1.60 & 1.90 & 2.07 & 2.09 & 1.08 & 1.12 & 1.10 & 1.10 & 0.93 & 0.86 & 0.83 & 0.83\\
Weighted ALR, $k = 2$ & 1.53 & 1.79 & 1.91 & 1.93 & 1.07 & 1.11 & 1.10 & 1.10 & 0.96 & 0.91 & 0.90 & 0.90\\
ALR, coarse bins & 1.50 & 1.76 & 1.86 & 1.87 & 1.07 & 1.11 & 1.08 & 1.09 & 0.96 & 0.92 & 0.92 & 0.92\\
\bottomrule
\end{tabular}
}\\
      \vspace{4mm}
      \subfloat[Differences between adjacent stages]{\footnotesize
\begin{tabular}{lrrrrrrrrr}
\toprule
\multicolumn{1}{c}{ } & \multicolumn{3}{c}{Age} & \multicolumn{3}{c}{Income} & \multicolumn{3}{c}{Height} \\
\cmidrule(l{3pt}r{3pt}){2-4} \cmidrule(l{3pt}r{3pt}){5-7} \cmidrule(l{3pt}r{3pt}){8-10}
 & \shortstack[c]{Pre \\ $-$ App} & \shortstack[c]{Ser \\ $-$ Pre} & \shortstack[c]{Pro \\ $-$ Ser} & \shortstack[c]{Pre \\ $-$ App} & \shortstack[c]{Ser \\ $-$ Pre} & \shortstack[c]{Pro \\ $-$ Ser} & \shortstack[c]{Pre \\ $-$ App} & \shortstack[c]{Ser \\ $-$ Pre} & \shortstack[c]{Pro \\ $-$ Ser}\\
\midrule
ALR, fine bins & 0.38 & 0.21 & 0.02 & 0.03 & -0.04 & 0.00 & -0.09 & -0.02 & 0.01\\
 & (0.00) & (0.00) & (1.00) & (0.01) & (0.91) & (1.00) & (0.00) & (1.00) & (1.00)\\
 & {}[0.00] & {}[0.00] & {}[0.78] & {}[0.00] & {}[0.58] & {}[0.99] & {}[0.00] & {}[1.00] & {}[1.00]\\
Weighted ALR, $k = 1$ & 0.30 & 0.17 & 0.02 & 0.04 & -0.01 & 0.00 & -0.07 & -0.03 & 0.00\\
 & (0.00) & (0.00) & (1.00) & (0.00) & (0.96) & (1.00) & (0.00) & (0.95) & (1.00)\\
 & {}[0.00] & {}[0.00] & {}[0.78] & {}[0.00] & {}[0.65] & {}[1.00] & {}[0.00] & {}[0.76] & {}[1.00]\\
Weighted ALR, $k = 2$ & 0.26 & 0.12 & 0.01 & 0.04 & -0.01 & 0.00 & -0.05 & -0.01 & 0.01\\
 & (0.00) & (0.00) & (1.00) & (0.00) & (0.75) & (1.00) & (0.00) & (1.00) & (1.00)\\
 & {}[0.00] & {}[0.00] & {}[0.78] & {}[0.00] & {}[0.43] & {}[0.99] & {}[0.00] & {}[0.99] & {}[1.00]\\
ALR, coarse bins & 0.26 & 0.10 & 0.02 & 0.03 & -0.02 & 0.01 & -0.04 & 0.00 & 0.00\\
 & (0.00) & (0.00) & (1.00) & (0.00) & (0.52) & (1.00) & (0.00) & (1.00) & (1.00)\\
 & {}[0.00] & {}[0.00] & {}[0.76] & {}[0.00] & {}[0.26] & {}[0.94] & {}[0.00] & {}[1.00] & {}[1.00]\\
\bottomrule
\end{tabular}
}
  \end{center}\footnotesize
  \textit{Note}: Panel (a) reports each index on the matrix that pools the pairs of 2021--2024. Panel (b) reports the difference between adjacent stages, the later stage minus the earlier one; stages are abbreviated App (Application), Pre (Pre-relationship), Ser (Serious relationship), and Pro (Proposal). Each test has the null hypothesis that the index takes the same value at the two adjacent stages. Under each difference, parentheses give the Romano--Wolf step-down p-value adjusted across all 36 hypotheses in the panel, and brackets give the p-value adjusted within the attribute (12 hypotheses). Standard errors come from B = 1,000 bootstrap replications that resample men and women independently and weight each pair by the product of its two draws; this scheme is a heuristic for many-to-many matchings, and the p-values are diagnostic. Weighted ALRs use $w_{ij} = \mathbf{1}\{|i - j| \le k\}$ on the fine bins. Coarse bins combine adjacent fine categories into five-year age groups, income groups labeled 300--400, 500--700, 800--1,000, and 1,100--2,100 (in 10,000 yen), and ten-centimeter height groups. The first age group is 23--24 and the first height group combines categories labeled 142--148 cm.
\end{table}

\subsubsection{Changes over time}\label{subsubsec:results_time}

Table \ref{tb:assortativeness_proposal_time} shows that sorting on age among engaged pairs increased between 2021 and 2024, the counterpart of the question that \citet{chiappori2025changes} ask about educational sorting in the United States. Panel (a) reports the levels by year, and panel (b) the change from 2021 to 2024 with p-values adjusted across the twelve hypotheses. For age, the weighted ALRs rise by 0.26 and 0.22 and the coarse-bin ALR by 0.13; these increases are distinguishable from zero. The one-year ALR shows no significant change. For income, all four indices show no significant change. For height, no change is distinguishable from zero. Table \ref{tb:assortativeness_stage_time} repeats the comparison at the three earlier stages, where the p-values are diagnostic: age sorting rises at the Application and Pre-relationship stages for all four indices, income sorting rises at the Application stage for all four, and no change in height sorting is distinguishable from zero at any stage. The increase in age sorting among engaged pairs is therefore already present at the stage of meeting requests, which suggests that it reflects a change in whom users approach and not only in whom they finally choose.

\begin{table}[!htbp]
  \setlength{\tabcolsep}{4pt}
  \begin{center}
      \caption{Changes in Assortativeness at the Proposal Stage, 2021 to 2024}
      \label{tb:assortativeness_proposal_time}
      \subfloat[Levels by year]{\footnotesize
\begin{tabular}{lrrrrrrrrrrrr}
\toprule
\multicolumn{1}{c}{ } & \multicolumn{4}{c}{Age} & \multicolumn{4}{c}{Income} & \multicolumn{4}{c}{Height} \\
\cmidrule(l{3pt}r{3pt}){2-5} \cmidrule(l{3pt}r{3pt}){6-9} \cmidrule(l{3pt}r{3pt}){10-13}
 & 2021 & 2022 & 2023 & 2024 & 2021 & 2022 & 2023 & 2024 & 2021 & 2022 & 2023 & 2024\\
\midrule
ALR, fine bins & 2.22 & 2.30 & 2.35 & 2.34 & 1.06 & 1.09 & 1.05 & 1.12 & 0.70 & 0.94 & 0.80 & 0.84\\
Weighted ALR, $k = 1$ & 1.90 & 2.12 & 2.12 & 2.16 & 1.07 & 1.10 & 1.09 & 1.13 & 0.84 & 0.87 & 0.80 & 0.83\\
Weighted ALR, $k = 2$ & 1.79 & 1.92 & 1.93 & 2.01 & 1.10 & 1.10 & 1.10 & 1.11 & 0.89 & 0.91 & 0.89 & 0.92\\
ALR, coarse bins & 1.80 & 1.86 & 1.87 & 1.94 & 1.06 & 1.11 & 1.10 & 1.09 & 0.92 & 0.92 & 0.93 & 0.92\\
\bottomrule
\end{tabular}
}\\
      \vspace{4mm}
      \subfloat[Change from 2021 to 2024]{\footnotesize
\begin{tabular}{lrrr}
\toprule
 & Age & Income & Height\\
\midrule
ALR, fine bins & 0.12 & 0.06 & 0.15\\
 & (0.85) & (0.82) & (0.59)\\
 & {}[0.34] & {}[0.51] & {}[0.34]\\
Weighted ALR, $k = 1$ & 0.26 & 0.06 & 0.00\\
 & (0.00) & (0.06) & (0.98)\\
 & {}[0.00] & {}[0.02] & {}[0.98]\\
Weighted ALR, $k = 2$ & 0.22 & 0.01 & 0.03\\
 & (0.00) & (0.85) & (0.85)\\
 & {}[0.00] & {}[0.51] & {}[0.63]\\
ALR, coarse bins & 0.13 & 0.02 & 0.00\\
 & (0.04) & (0.82) & (0.98)\\
 & {}[0.01] & {}[0.51] & {}[0.98]\\
\bottomrule
\end{tabular}
}
  \end{center}\footnotesize
  \textit{Note}: Panel (a) reports each index on the engaged pairs of each year. Panel (b) reports the change from 2021 to 2024; each test has the null hypothesis that the index among engaged pairs is the same in 2021 and 2024. Parentheses give the Romano--Wolf step-down p-value adjusted across the 12 hypotheses in the panel, and brackets give the p-value adjusted within the attribute (4 hypotheses). Standard errors come from B = 1,000 bootstrap replications that resample engaged couples with replacement within each year. Indices and bins are defined as in Table \ref{tb:assortativeness_by_stage}.
\end{table}

\begin{table}[!htbp]
  \setlength{\tabcolsep}{4pt}
  \begin{center}
      \caption{Changes in Assortativeness at the Application, Pre-relationship, and Serious Relationship Stages, 2021 to 2024}
      \label{tb:assortativeness_stage_time}
      {\footnotesize
\begin{tabular}{lrrrrrrrrr}
\toprule
\multicolumn{1}{c}{ } & \multicolumn{3}{c}{Age} & \multicolumn{3}{c}{Income} & \multicolumn{3}{c}{Height} \\
\cmidrule(l{3pt}r{3pt}){2-4} \cmidrule(l{3pt}r{3pt}){5-7} \cmidrule(l{3pt}r{3pt}){8-10}
 & App & Pre & Ser & App & Pre & Ser & App & Pre & Ser\\
\midrule
ALR, fine bins & 0.13 & 0.29 & 0.27 & 0.04 & 0.06 & 0.07 & 0.02 & 0.01 & -0.05\\
 & (0.00) & (0.00) & (0.75) & (0.00) & (0.16) & (0.98) & (0.18) & (1.00) & (1.00)\\
 & {}[0.00] & {}[0.00] & {}[0.11] & {}[0.00] & {}[0.05] & {}[0.75] & {}[0.11] & {}[1.00] & {}[1.00]\\
Weighted ALR, $k = 1$ & 0.13 & 0.26 & 0.25 & 0.03 & 0.03 & 0.03 & 0.01 & 0.00 & 0.00\\
 & (0.00) & (0.00) & (0.11) & (0.00) & (0.05) & (0.98) & (0.18) & (1.00) & (1.00)\\
 & {}[0.00] & {}[0.00] & {}[0.02] & {}[0.00] & {}[0.02] & {}[0.75] & {}[0.11] & {}[1.00] & {}[1.00]\\
Weighted ALR, $k = 2$ & 0.12 & 0.23 & 0.22 & 0.02 & 0.02 & 0.00 & 0.01 & 0.00 & -0.02\\
 & (0.00) & (0.00) & (0.01) & (0.00) & (0.18) & (1.00) & (0.10) & (1.00) & (1.00)\\
 & {}[0.00] & {}[0.00] & {}[0.00] & {}[0.00] & {}[0.06] & {}[0.99] & {}[0.05] & {}[1.00] & {}[1.00]\\
ALR, coarse bins & 0.10 & 0.16 & 0.14 & 0.02 & 0.03 & 0.00 & 0.00 & 0.00 & 0.00\\
 & (0.00) & (0.00) & (0.29) & (0.00) & (0.29) & (1.00) & (0.97) & (1.00) & (1.00)\\
 & {}[0.00] & {}[0.00] & {}[0.05] & {}[0.00] & {}[0.10] & {}[0.99] & {}[0.91] & {}[1.00] & {}[1.00]\\
\bottomrule
\end{tabular}
}
  \end{center}\footnotesize
  \textit{Note}: Each cell reports the change in the index from 2021 to 2024 at the stage indicated; stages are abbreviated as in Table \ref{tb:assortativeness_by_stage}. Each test has the null hypothesis that the index at the stage is the same in 2021 and 2024. Parentheses give the Romano--Wolf step-down p-value adjusted across all 36 hypotheses in the table, and brackets give the p-value adjusted within the attribute (12 hypotheses). Standard errors come from B = 1,000 bootstrap replications that resample men and women independently and weight each pair by the product of its two draws; the p-values are diagnostic. Indices and bins are defined as in Table \ref{tb:assortativeness_by_stage}.
\end{table}

\subsubsection{Comparison with the indices of Chiappori et al. (2025)}\label{subsubsec:results_comparison}

Table \ref{tb:comparison_chiappori} compares our indices with the seven indices of \citet{chiappori2025changes} on the change in age sorting among engaged pairs, using five age groups in place of their five education groups. Panel (a) reports the matching matrices of 2021 and 2024 with their marginal distributions. Most of the mass lies on the diagonal and in the cells just below it, in which the man is one group older than the woman, and the marginals shift between the two years: the share of women aged 23--29 rises from 14.3 to 16.0 percent, and the share of men aged 45--59 falls from 23.8 to 18.7 percent. Panel (b) reports the change in each index from 2021 to 2024 for the four two-by-two submarkets that \citet{chiappori2025changes} consider and for the full five-group matrix, with p-values adjusted across all 34 hypotheses. On the full matrix, the two indices defined for it both increase: the normalized trace, whose increase is distinguishable from zero, and the weighted $L$ ratio, which coincides with the ALR on five groups and whose increase is distinguishable from zero only when the p-value is adjusted within the column. Our weighted ALRs on one-year bins also increase and are distinguishable from zero under the same 34-hypothesis adjustment. Within submarkets the indices disagree, as in \citet{chiappori2025changes}. In submarket (1), which compares users aged 40 or older with those aged 35--39, the $L_1$ ratio, the weighted $L$ ratio, the normalized trace, and the $\chi^2$ index increase, whereas the odds ratio and the $L_2$ ratio are not distinguishable from zero. In submarket (4), which compares users aged 40 or older with everyone younger, the normalized trace, the $\chi^2$ index, and the $L_1$ ratio increase, whereas the weighted $L$ ratio shows no significant change. In this comparison, the odds ratio and the type-specific likelihood ratios are evaluated on two-type submarkets, and the change in a type-specific ratio depends on the type it conditions on; the ALR and its weighted versions are defined on the full matrix and give one answer for it, which is why the application reports them.

\begin{table}[!htbp]
  \setlength{\tabcolsep}{3pt}
  \renewcommand{\arraystretch}{0.88}
  \captionsetup[sub]{aboveskip=2pt,belowskip=2pt}
  \begin{center}
      \caption{Changes in Age Assortativeness at the Proposal Stage, 2021 to 2024: Comparison with the Indices of Chiappori et al. (2025)}
      \label{tb:comparison_chiappori}
      \subfloat[Matching matrices of engaged pairs by five age groups (cell percentages)]{\footnotesize
\begin{tabular}{lrrrrrrrrrrrr}
\toprule
\multicolumn{1}{c}{ } & \multicolumn{6}{c}{2021 (N = 5,154)} & \multicolumn{6}{c}{2024 (N = 7,382)} \\
\cmidrule(l{3pt}r{3pt}){2-7} \cmidrule(l{3pt}r{3pt}){8-13}
Men / Women & 23--29 & 30--34 & 35--39 & 40--44 & 45--59 & \shortstack[r]{Marginal: \\ men} & 23--29 & 30--34 & 35--39 & 40--44 & 45--59 & \shortstack[r]{Marginal: \\ men}\\
\midrule
23--29 & 4.1 & 1.4 & 0.1 & 0.0 & 0.0 & 5.5 & 5.0 & 2.0 & 0.1 & 0.0 & 0.0 & 7.1\\
30--34 & 7.2 & 11.9 & 2.0 & 0.1 & 0.0 & 21.2 & 7.9 & 14.6 & 2.3 & 0.1 & 0.0 & 25.0\\
35--39 & 2.6 & 13.9 & 11.1 & 1.1 & 0.0 & 28.6 & 2.6 & 13.3 & 12.7 & 1.1 & 0.0 & 29.7\\
40--44 & 0.3 & 3.9 & 11.2 & 4.9 & 0.4 & 20.8 & 0.4 & 2.8 & 10.7 & 5.2 & 0.4 & 19.5\\
45--59 & 0.1 & 0.8 & 5.7 & 8.2 & 9.0 & 23.8 & 0.0 & 0.4 & 3.2 & 7.3 & 7.8 & 18.7\\
Marginal: women & 14.3 & 31.9 & 30.0 & 14.3 & 9.5 & 100.0 & 16.0 & 33.1 & 28.9 & 13.7 & 8.3 & 100.0\\
\bottomrule
\end{tabular}
}\\
      \subfloat[Change from 2021 to 2024 by index]{\footnotesize
\begin{tabular}{lrrrrr}
\toprule
\multicolumn{1}{c}{ } & \multicolumn{4}{c}{Submarkets} & \multicolumn{1}{c}{Global} \\
\cmidrule(l{3pt}r{3pt}){2-5} \cmidrule(l{3pt}r{3pt}){6-6}
 & \shortstack[c]{(1) 40--59 \\ vs. 35--39} & \shortstack[c]{(2) 45--59 \\ vs. 40--44} & \shortstack[c]{(3) 40--44 \\ vs. 35--39} & \shortstack[c]{(4) 40--59 \\ vs. 23--39} & \shortstack[c]{(5) All \\ five groups}\\
\midrule
\multicolumn{6}{l}{\textbf{Chiappori et al. (2025) indices, five age groups}}\\
\hspace{1em}1. Odds ratio & 3.10 & 1.59 & 1.06 & 10.22 & \\
\hspace{1em} & (0.88) & (0.99) & (0.90) & (0.86) & \\
\hspace{1em} & {}[0.47] & {}[0.80] & {}[0.48] & {}[0.35] & \\
\hspace{1em}2. $L_1$ ratio (older group) & 0.08 & 0.05 & 0.10 & 0.34 & \\
\hspace{1em} & (0.00) & (0.49) & (0.47) & (0.00) & \\
\hspace{1em} & {}[0.00] & {}[0.25] & {}[0.17] & {}[0.00] & \\
\hspace{1em}3. $L_2$ ratio (younger group) & -0.01 & -0.04 & 0.01 & -0.03 & \\
\hspace{1em} & (0.99) & (0.97) & (0.99) & (0.16) & \\
\hspace{1em} & {}[0.84] & {}[0.80] & {}[0.65] & {}[0.04] & \\
\hspace{1em}4. Weighted $L$ ratio & 0.08 & 0.04 & 0.02 & -0.02 & 0.12\\
\hspace{1em} & (0.00) & (0.75) & (0.92) & (0.77) & (0.07)\\
\hspace{1em} & {}[0.00] & {}[0.43] & {}[0.54] & {}[0.35] & {}[0.01]\\
\hspace{1em}5. Normalized trace & 0.04 & 0.01 & 0.03 & 0.04 & 0.04\\
\hspace{1em} & (0.04) & (0.97) & (0.33) & (0.00) & (0.00)\\
\hspace{1em} & {}[0.01] & {}[0.80] & {}[0.11] & {}[0.00] & {}[0.00]\\
\hspace{1em}6. $\chi^2$ & 0.05 & 0.02 & 0.02 & 0.05 & \\
\hspace{1em} & (0.04) & (0.92) & (0.75) & (0.01) & \\
\hspace{1em} & {}[0.01] & {}[0.69] & {}[0.32] & {}[0.00] & \\
\hspace{1em}7. Minimum distance & 0.01 & 0.02 & 0.04 & 0.00 & \\
\hspace{1em} & (0.99) & (0.99) & (0.97) & (0.99) & \\
\hspace{1em} & {}[0.84] & {}[0.80] & {}[0.65] & {}[0.89] & \\
\multicolumn{6}{l}{\textbf{This paper, one-year age bins}}\\
\hspace{1em}ALR, fine bins &  &  &  &  & 0.12\\
\hspace{1em} &  &  &  &  & (0.92)\\
\hspace{1em} &  &  &  &  & {}[0.34]\\
\hspace{1em}Weighted ALR, $k = 1$ &  &  &  &  & 0.26\\
\hspace{1em} &  &  &  &  & (0.00)\\
\hspace{1em} &  &  &  &  & {}[0.00]\\
\hspace{1em}Weighted ALR, $k = 2$ &  &  &  &  & 0.22\\
\hspace{1em} &  &  &  &  & (0.00)\\
\hspace{1em} &  &  &  &  & {}[0.00]\\
\hspace{1em}ALR, coarse bins &  &  &  &  & 0.13\\
\hspace{1em} &  &  &  &  & (0.05)\\
\hspace{1em} &  &  &  &  & {}[0.01]\\
\bottomrule
\end{tabular}
}
  \end{center}\footnotesize
  \textit{Note}: Panel (a): cell percentages of engaged pairs by the man's (rows) and the woman's (columns) age group; the last row and column are the marginal distributions. Panel (b): change in each index from 2021 to 2024 in the two-by-two submarkets (1)--(4), older group first, and in the full five-group matrix (5). Rows 1--7 follow Table 3 of \citet{chiappori2025changes}, with $L_1$ conditioning on the older group and $L_2$ on the younger; the weighted $L$ ratio equals the ALR of the matrix. The last four rows are the indices of Table \ref{tb:assortativeness_proposal_time}. Null hypothesis of each test: no change from 2021 to 2024. Parentheses: step-down p-values adjusted across all 34 hypotheses; brackets: adjusted within the column (rows 1--7) or within the four indices of this paper. All indices use the same B = 1,000 couple-bootstrap draws within year.
\end{table}

\FloatBarrier
\section{Conclusion}
\label{sec:conclusion}

This paper revisits the axiomatic foundation for measuring assortative matching. We show that the axioms of \citet{chiappori2025changes} do not characterize the aggregate likelihood ratio, recover the characterization in two ways, axiomatize weighted ALRs, and extend ALR to many-to-many matchings with an explicit benchmark, an approximation bound, and an axiomatization. Using platform data that follow pairs from application to engagement and record profile attributes, including document-checked age and income, we find that sorting on age is stronger than sorting on income or height and that it strengthens in the early stages of partner selection, before couples commit to marriage. Because our measures require only data on who matches with whom and on agents' attributes, without specifying preferences, transfers, or an equilibrium concept, the same stage-by-stage comparison applies to labor markets, buyer--seller networks, and other markets in which many-to-many contacts narrow down to final matches.

\bibliographystyle{apalike}
\bibliography{ibj_project}

\newpage
\appendix
\section*{Appendices}
\addcontentsline{toc}{section}{Appendices}

\section{Errors in Chiappori et al.'s (2025) other axiomatizations}

This appendix records errors in the axiomatizations of the odds ratio and the normalized trace.

\subsection{Odds ratio}
\label{sect-odds}

This subsection explains why the odds-ratio characterization in \citet{chiappori2025changes} does not determine rankings on the boundary of the two-type domain and states a corrected version.

Let
\[
\mathcal X_2:=\R_{\geq 0}^4\setminus\{0\},
\qquad
\mathcal X_{2,>0}:=\{M\in\mathcal X_2:M>0\}.
\]
For $M=(a,b,c,d)\in\mathcal X_2$, define the odds ratio by
\[
I_O(M)=\begin{cases}
\dfrac{ad}{bc}, & b\neq 0\text{ and }c\neq 0,\\[0.6em]
+\infty, & bc=0.
\end{cases}
\]
A preorder $\succeq$ is a reflexive and transitive binary relation on $\mathcal X_2$. We revisit the following axioms.

\begin{itemize}
\item \textbf{Marginal Independence.} For any $M\in\mathcal X_2$ and $\lambda>0$,
\[
\begin{pmatrix}a&b\\c&d\end{pmatrix}
\sim
\begin{pmatrix}\lambda a&\lambda b\\c&d\end{pmatrix}
\sim
\begin{pmatrix}a&b\\\lambda c&\lambda d\end{pmatrix}
\sim
\begin{pmatrix}\lambda a&b\\\lambda c&d\end{pmatrix}
\sim
\begin{pmatrix}a&\lambda b\\c&\lambda d\end{pmatrix}.
\]

\item \textbf{Maximum Homogamy.} For any $M\in\mathcal X_2$ and any $(a,b,c,d)\in\mathcal X_2$ with $ad>0$ and $bc=0$, $(a,b,c,d)\succeq M$.

\item \textbf{Marginal Monotonicity.} For every pair of positive matching matrices $M=(a,b,c,d)$ and $M'=(a',b',c',d')$ in $\mathcal X_{2,>0}$ with the same marginal distributions,
\[
M\succ M'\quad\Longleftrightarrow\quad a>a'.
\]
\end{itemize}

Theorem 1 of \citet{chiappori2025changes} states that the odds ratio induces the unique preorder satisfying Marginal Monotonicity, Marginal Independence, and Maximum Homogamy. This statement is not correct on the full domain.

Define three subsets of $\mathcal X_2$:
\[
\mathcal X_\infty:=\{(a,b,c,d):bc=0,\ ad>0\},
\]
\[
\widehat{\mathcal X}_\infty:=\{(a,b,c,d):bc=0,\ ad=0\},
\]
\[
\mathcal X_0:=\{(a,b,c,d):bc>0,\ ad=0\}.
\]
Then $\mathcal X_2=\mathcal X_\infty\cup\widehat{\mathcal X}_\infty\cup\mathcal X_0\cup\mathcal X_{2,>0}$. The odds ratio assigns $+\infty$ to matrices in $\mathcal X_\infty\cup\widehat{\mathcal X}_\infty$, assigns $0$ to matrices in $\mathcal X_0$, and assigns a positive real number to matrices in $\mathcal X_{2,>0}$.

Define a preorder $\succeq^*$ as follows:
\begin{enumerate}[label=(\arabic*)]
\item For each $M\in\mathcal X_\infty$ and $M'\in\widehat{\mathcal X}_\infty\cup\mathcal X_0\cup\mathcal X_{2,>0}$, $M\succ^* M'$. For each $M,M'\in\mathcal X_\infty$, $M\sim^* M'$.
\item For each $M\in\widehat{\mathcal X}_\infty$ and $M'\in\mathcal X_0\cup\mathcal X_{2,>0}$, $M\succ^* M'$. For each $M,M'\in\widehat{\mathcal X}_\infty$, $M\sim^* M'$.
\item For each $M\in\mathcal X_0$ and $M'\in\mathcal X_{2,>0}$, $M\succ^* M'$. For each $M,M'\in\mathcal X_0$, $M\sim^* M'$.
\item For each $M,M'\in\mathcal X_{2,>0}$, $M\succeq^* M'$ if and only if $I_O(M)\geq I_O(M')$.
\end{enumerate}

The preorder $\succeq^*$ satisfies all three axioms above. Marginal Monotonicity has bite only on $\mathcal X_{2,>0}$, where $\succeq^*$ is induced by the odds ratio. Maximum Homogamy follows from condition (1). Marginal Independence holds because scaling any row or column by a positive factor keeps a matrix within the same one of the four sets above, and on the positive domain the odds ratio itself is invariant to those operations.

However, $\succeq^*$ is not induced by the odds ratio. For example, let
\[
M=\begin{pmatrix}1&0\\0&1\end{pmatrix},
\qquad
M'=\begin{pmatrix}1&0\\0&0\end{pmatrix}.
\]
Then $M\in\mathcal X_\infty$ and $M'\in\widehat{\mathcal X}_\infty$, so $M\succ^* M'$. But $I_O(M)=+\infty=I_O(M')$. Alternatively, let
\[
M=\begin{pmatrix}0&1\\1&0\end{pmatrix},
\qquad
M'=\begin{pmatrix}1&1\\1&1\end{pmatrix}.
\]
Then $M\in\mathcal X_0$ and $M'\in\mathcal X_{2,>0}$, so $M\succ^* M'$, whereas $I_O(M)=0<1=I_O(M')$.

The original proof correctly shows that any preorder satisfying the axioms must coincide with the preorder induced by the odds ratio on $\mathcal X_{2,>0}$. The problem is that the axioms do not determine the ranking on the remaining domain. To remedy this, we strengthen Maximum Homogamy and add a heterogamy condition.\footnote{\citet{chiappori2025changes} distinguish between the two modifiers \lq \lq full'' and \lq \lq maximum'' in their homogamy axioms as follows: ``full'' refers to comparisons with matching matrices in which {\it all} off-diagonal entries are zero, whereas ``maximum'' refers more broadly to comparisons with matching matrices in which at least one off-diagonal entry is zero. Our use of the terms ``full'' and ``maximum'' follows their terminology.}

\begin{itemize}
\item \textbf{Maximum Homogamy$+$.} For any $M\in\mathcal X_2$ and any $(a,b,c,d)\in\mathcal X_2$ such that $bc=0$, $(a,b,c,d)\succeq M$.

\item \textbf{Maximum Heterogamy.} For each $M\in\mathcal X_2$ and each $(a,b,c,d)\in\mathcal X_2$ such that $b>0$, $c>0$, and $ad=0$, $M\succeq(a,b,c,d)$.
\end{itemize}

\begin{proposition}[Corrected odds-ratio characterization]
The odds ratio induces the unique preorder that satisfies Marginal Monotonicity, Marginal Independence, Maximum Homogamy$+$, and Maximum Heterogamy.
\end{proposition}

\begin{proof}
Marginal Monotonicity and Marginal Independence imply that the preorder agrees with the odds-ratio preorder on $\mathcal X_{2,>0}$. Maximum Homogamy$+$ makes all matrices with $bc=0$ weakly maximal, eliminating the distinction between $\mathcal X_\infty$ and $\widehat{\mathcal X}_\infty$. Maximum Heterogamy makes all matrices in $\mathcal X_0$ weakly minimal, so they are indifferent to one another. Together with Marginal Monotonicity, this pins down the strict comparison between positive matrices and these extreme heterogamy matrices: if $M\in\mathcal X_{2,>0}$ and $H\in\mathcal X_0$, then $M\succ H$. Otherwise, $M\sim H$ would imply, for sufficiently small $\eps>0$, that $M\succ M_\eps$ for $M_\eps=(a-\eps,b+\eps,c+\eps,d-\eps)\in\mathcal X_{2,>0}$ with the same marginals, contradicting Maximum Heterogamy applied to $M_\eps$ and $H$. Symmetrically, if a matrix $B$ with $bc=0$ satisfied $B\sim M$ for some $M\in\mathcal X_{2,>0}$, then $M'=(a+\eps,b-\eps,c-\eps,d+\eps)\in\mathcal X_{2,>0}$ would satisfy $M'\succ M\sim B$ by Marginal Monotonicity, contradicting Maximum Homogamy$+$; hence $B\succ M$, and by transitivity $B\succ H$ for every $H\in\mathcal X_0$. Thus all boundary rankings agree with the odds ratio, and uniqueness follows.
\end{proof}

\subsection{Normalized trace}

This subsection gives an analogous counterexample for the normalized trace and points to the weighted normalized-trace characterization in Appendix \ref{app:weighted-normalized-trace}.

\citet{chiappori2025changes} also consider the normalized trace
\[
I_{\tr}(M)=\begin{cases}
1, & bc=0,\\[0.4em]
\dfrac{a+d}{a+b+c+d}, & abcd\neq 0,\\[0.8em]
0, & ad=0.
\end{cases}
\]
This formula is not well-defined on the entire domain $\mathcal X_2$, because the conditions $bc=0$ and $ad=0$ overlap. For example, $(1,1,0,0)$ satisfies both clauses. We therefore restrict attention to
\[
\widehat{\mathcal X}_2:=\{M\in\mathcal X_2:bc\neq 0\text{ or }ad\neq 0\}.
\]
On this restricted domain, the overlapping case $bc=0$ and $ad=0$ is excluded.

Theorem 2 of \citet{chiappori2025changes} states that the normalized trace is the unique index, up to positive affine transformation, satisfying the three invariance axioms, Marginal Monotonicity, Maximum Homogamy, and Population Decomposability. The last two axioms are:

\begin{itemize}
\item \textbf{Maximum Homogamy.} For each $M\in\widehat{\mathcal X}_2$ and each $(a,b,c,d)\in\widehat{\mathcal X}_2$ such that $ad>0$ and $bc=0$, $I(a,b,c,d)\geq I(M)$.

\item \textbf{Population Decomposability.} For each pair of positive matchings $M,M'\in\widehat{\mathcal X}_2$,
\[
I(M+M')=\frac{|M|}{|M|+|M'|}I(M)+\frac{|M'|}{|M|+|M'|}I(M').
\]
\end{itemize}

The theorem does not hold as stated. Consider the modified index
\[
I'_{\tr}(M)=\begin{cases}
1, & bc=0,\\[0.4em]
\dfrac{a+d+\frac12 b+\frac12 c}{a+b+c+d}, & abcd\neq 0,\\[0.8em]
0, & ad=0.
\end{cases}
\]
One can verify that $I'_{\tr}$ satisfies the three invariance axioms, Marginal Monotonicity, Maximum Homogamy, and Population Decomposability. It is not a positive affine transformation of $I_{\tr}$. If $bc=0$ and $ad>0$, then $I_{\tr}=I'_{\tr}=1$, so any affine relation $I'_{\tr}=\alpha I_{\tr}+\beta$ requires $\alpha+\beta=1$. If $ad=0$ and $bc>0$, then $I_{\tr}=I'_{\tr}=0$, so $\beta=0$. Thus $\alpha=1$, which would require $I'_{\tr}=I_{\tr}$ everywhere. But for $M=(1,1,1,1)$,
\[
I_{\tr}(M)=\frac12,
\qquad
I'_{\tr}(M)=\frac34.
\]
This contradiction shows that the normalized-trace axiomatization is not correct as stated. Appendix \ref{app:weighted-normalized-trace} gives a multi-type linear-score characterization of weighted normalized-trace indices.

\section{Multi-type odds ratio}

In this appendix we offer a multi-type generalization of the odds ratio. Let $T=\{1,\ldots,n\}$ and let $\Theta:=T\times T$ denote the set of type profiles, with generic element $\tau=(t,t')$. A matching $M=(m_\tau)_{\tau\in\Theta}$ is an element of $\R_{>0}^{\Theta}$; throughout this appendix we restrict attention to positive matrices. An index is a function $I:\R_{>0}^{\Theta}\to\R$. We write $M\succeq M'$ when $I(M)\geq I(M')$.

For $M=(m_\tau)\in\R_{>0}^{\Theta}$, $\theta\in\Theta$, and $m'_\theta>0$, let $(M_{-\theta},m'_\theta)$ denote the matching obtained from $M$ by replacing the $\theta$-th entry with $m'_\theta$, while keeping all other entries fixed.

\begin{itemize}
\item \textbf{Cell Scale Independence.} Let $M=(m_\tau)$ and $\widetilde M=(\widetilde m_\tau)$ be positive. For each $\theta\in\Theta$ and $\lambda>0$,
\[
I(M)>I(\widetilde M)
\quad\Longleftrightarrow\quad
I(M_{-\theta},\lambda m_\theta)>I(\widetilde M_{-\theta},\lambda\widetilde m_\theta).
\]
\end{itemize}
This axiom is closely related to a separability condition used by \citet[p.~1007]{osborne1976irrelevant} in the analysis of social welfare functions. It states that if, in both matchings, the count in any single cell is multiplied by the same positive factor, then the index preserves the ranking of assortativeness between the two matchings.

For an index $I:\R_{>0}^{\Theta}\to\R$, Scale Invariance means $I(\lambda M)=I(M)$ for every $\lambda>0$. Continuity means that $I$ is continuous on $\R_{>0}^{\Theta}$. A vector $Z=(z_\tau)_{\tau\in\Theta}\in\R^{\Theta}$ has zero total sum if $\sum_{\tau\in\Theta}z_\tau=0$.

\begin{theorem}[Multi-type odds-ratio form]\label{thm:multi-odds}
An index $I:\R_{>0}^{\Theta}\to\R$ satisfies Scale Invariance, Continuity, and Cell Scale Independence if and only if there are a continuous and strictly increasing function $\xi:\R_{++}\to\R$ and a vector $Z=(z_\tau)_{\tau\in\Theta}\in\R^{\Theta}$ with zero total sum such that
\[
I(M)=\xi\!\left(\prod_{\tau\in\Theta}m_\tau^{z_\tau}\right)
\qquad\text{for all }M=(m_\tau)\in\R_{>0}^{\Theta}.
\]
\end{theorem}

\begin{proof}
The ``if'' direction is immediate except for two checks. If $I$ has the displayed form and $\sum_\tau z_\tau=0$, then
\[
\prod_\tau(\lambda m_\tau)^{z_\tau}
=\lambda^{\sum_\tau z_\tau}\prod_\tau m_\tau^{z_\tau}
=\prod_\tau m_\tau^{z_\tau},
\]
so Scale Invariance holds. Cell Scale Independence holds because multiplying the same coordinate $\theta$ in both matrices by $\lambda$ adds the same term $z_\theta\log\lambda$ to both log-products:
\[
\sum_\tau z_\tau\log m_\tau
\quad\text{and}\quad
\sum_\tau z_\tau\log \widetilde m_\tau.
\]
Since $\xi$ and $\log$ are strictly increasing, the ranking is unchanged.

For the converse, we use three claims. For a coordinate $\theta\in\Theta$, say that $I$ is nondecreasing in $\theta$ if replacing $m_\theta$ by a larger value never lowers $I$, and nonincreasing in $\theta$ if replacing $m_\theta$ by a larger value never raises $I$.

\emph{Claim 1.} For each $\theta\in\Theta$, $I$ is either nondecreasing or nonincreasing in $\theta$.

Suppose not. Then there are $M,m'_{\theta}>m_\theta$ and $\widetilde M,\widetilde m'_{\theta}>\widetilde m_\theta$ such that
\[
I(M)>I(M_{-\theta},m'_\theta),
\qquad
I(\widetilde M)<I(\widetilde M_{-\theta},\widetilde m'_\theta).
\]
By applying Cell Scale Independence coordinate by coordinate to all entries other than $\theta$, both inequalities can be reduced to inequalities involving only the scalar function
\[
f(x):=I(\one_{-\theta},x),\qquad x>0,
\]
where $\one$ is the all-one matching. Thus there exist $m_\theta<m'_\theta$ and $\widetilde m_\theta<\widetilde m'_\theta$ such that
\[
f(m_\theta)>f(m'_\theta),
\qquad
f(\widetilde m_\theta)<f(\widetilde m'_\theta).
\]
Let $m^*_\theta$ be the largest point in $[m_\theta,m'_\theta]$ at which $f$ equals $f(m_\theta)$, and define $\widetilde m^*_\theta$ analogously using $f(\widetilde m_\theta)$. These points exist by continuity and compactness, and the strict endpoint inequalities imply that neither is the right endpoint. The intermediate value theorem then gives $f(x)<f(m^*_\theta)$ for every $x\in(m^*_\theta,m'_\theta]$, and $f(x)>f(\widetilde m^*_\theta)$ for every $x\in(\widetilde m^*_\theta,\widetilde m'_\theta]$. Let
\[
\alpha_0:=\min\left\{\frac{m'_\theta}{m^*_\theta},\frac{\widetilde m'_\theta}{\widetilde m^*_\theta}\right\}>1.
\]
Then
\[
f(m^*_\theta)>f(\alpha_0 m^*_\theta),
\qquad
f(\widetilde m^*_\theta)<f(\alpha_0\widetilde m^*_\theta),
\]
contradicting Cell Scale Independence, because the second pair is obtained from the first comparison by multiplying the $\theta$ coordinate in both matrices by the same positive factor. This proves Claim 1.

Let
\[
\Theta_-:=\{\theta\in\Theta:I\text{ is not nondecreasing in }\theta\}.
\]
For $M=(m_\tau)$, define $M^{\mathrm{inv}}$ by
\[
(M^{\mathrm{inv}})_\tau=\begin{cases}
1/m_\tau, & \tau\in\Theta_-,\\
m_\tau, & \tau\notin\Theta_-.
\end{cases}
\]
and define $I^{\mathrm{inv}}(M):=I(M^{\mathrm{inv}})$.

\emph{Claim 2.} The index $I^{\mathrm{inv}}$ is nondecreasing in every coordinate.

If $I$ is nondecreasing in a coordinate, the claim is immediate. If $\theta\in\Theta_-$, Claim 1 implies that $I$ is nonincreasing in $\theta$; after replacing $m_\theta$ by $1/m_\theta$, this coordinate becomes nondecreasing for $I^{\mathrm{inv}}$.

\emph{Claim 3.} The index $I^{\mathrm{inv}}$ satisfies Cell Scale Independence.

For $\theta\notin\Theta_-$ this is immediate. For $\theta\in\Theta_-$, multiplying $m_\theta$ by $\lambda$ before inversion is the same as multiplying $1/m_\theta$ by $1/\lambda$ after inversion. Cell Scale Independence of $I$ therefore implies Cell Scale Independence of $I^{\mathrm{inv}}$.

We use an auxiliary lemma adapted from \citet[p.~1007]{osborne1976irrelevant}. An index that is nondecreasing in every coordinate and satisfies Cell Scale Independence can be written as
\[
J(M)=\xi\!\left(\prod_{\tau\in\Theta}m_\tau^{z_\tau}\right)
\]
for some strictly increasing $\xi:\R_{++}\to\R$ and some nonnegative exponent vector $Z=(z_\tau)_{\tau\in\Theta}\in\R_{\geq 0}^{\Theta}$. Applying this lemma to $I^{\mathrm{inv}}$, and then undoing the inversion, gives
\[
I(M)=\xi\!\left(\prod_{\tau\in\Theta}m_\tau^{\widehat z_\tau}\right),
\qquad
\widehat z_\tau:=
\begin{cases}
-z_\tau,&\tau\in\Theta_-,\\
z_\tau,&\tau\notin\Theta_-.
\end{cases}
\]
Thus $I$ has the stated multiplicative form with an unrestricted real exponent vector, which we again denote by $Z$.

If $z_\theta\ne0$ for some coordinate, then $\xi(t)=I(\one_{-\theta},t^{1/z_\theta})$ for $t>0$, so continuity of $I$ implies continuity of $\xi$. If all exponents vanish, $I$ is constant and we may choose the continuous, strictly increasing function $\xi(t)=I(\one)+\log t$.

It remains to show that the exponent vector has zero total sum. Let $\bar z:=\sum_{\tau\in\Theta}z_\tau$. For the all-one matrix $\one$ and $2\one$,
\[
I(\one)=\xi(1),
\qquad
I(2\one)=\xi(2^{\bar z}).
\]
By Scale Invariance, $I(\one)=I(2\one)$. Since $\xi$ is strictly increasing, $1=2^{\bar z}$, so $\bar z=0$.
\end{proof}

\section{Weighted normalized trace}\label{app:weighted-normalized-trace}

This appendix gives a linear-score characterization of weighted normalized-trace-type indices. Let
\[
\Nn:=\R_{\geq 0}^{n\times n}\setminus\{0\}.
\]
For $M=(m_{ij})\in\Nn$, let
\[
|M|:=\sum_{i=1}^n\sum_{j=1}^n m_{ij},
\]
and let $E^{ij}$ denote the matrix whose $(i,j)$ entry is one and whose other entries are zero. An assortativeness index is a function $I:\Nn\to\R$.

\begin{itemize}
\item \textbf{Continuity.} $I$ is continuous on $\Nn$.
\item \textbf{Population Decomposability.} For every $M,M'\in\Nn$,
\[
|M+M'|I(M+M')=|M|I(M)+|M'|I(M').
\]
Equivalently, the index of a pooled population is the population-size-weighted average of the two subpopulation indices.
\end{itemize}

\begin{theorem}[Linear-score representation]\label{thm:linear-score}
Suppose $I:\Nn\to\R$ satisfies Continuity and Population Decomposability. Then there exist constants $(w_{ij})_{i,j=1}^n\in\R^{n\times n}$ such that
\[
I(M)=\frac{\sum_{i=1}^n\sum_{j=1}^n w_{ij}m_{ij}}{|M|}
\qquad(M\in\Nn).
\]
Moreover, $w_{ij}=I(E^{ij})$ for every $(i,j)$.
\end{theorem}

\begin{proof}
Define
\[
J(M):=|M|I(M)\qquad(M\in\Nn).
\]
Population Decomposability says that $J(M+M')=J(M)+J(M')$ for all $M,M'\in\Nn$. Hence $J$ is additive on the positive cone.

First consider matrices with rational entries. By additivity, $J(kM)=kJ(M)$ for every positive integer $k$. Also $J(M)=kJ(M/k)$, so $J(M/k)=J(M)/k$. Thus $J(qM)=qJ(M)$ for every positive rational $q$. For a rational matrix $M=(m_{ij})$, additivity gives
\[
J(M)=\sum_{i,j}m_{ij}J(E^{ij}).
\]
Set $w_{ij}:=J(E^{ij})=I(E^{ij})$, since $|E^{ij}|=1$.

Now take an arbitrary $M\in\Nn$. Choose rational matrices $M_k\to M$ with $M_k\in\Nn$. Since $I$ and $|\cdot|$ are continuous on $\Nn$, $J$ is continuous at $M$. Therefore,
\[
J(M)=\lim_{k\to\infty}J(M_k)=\lim_{k\to\infty}\sum_{i,j}w_{ij}m_{ij}^k=\sum_{i,j}w_{ij}m_{ij}.
\]
Dividing by $|M|$ gives the desired representation.
\end{proof}

\begin{remark}
Under the representation in Theorem \ref{thm:linear-score}, Scale Invariance is automatic:
\[
I(\lambda M)=\frac{\sum_{i,j}w_{ij}\lambda m_{ij}}{\lambda |M|}=I(M).
\]
Thus, for normalized-trace-type indices, Scale Invariance is redundant once Population Decomposability is imposed.
\end{remark}

The following restrictions connect the weights to symmetry and monotonicity. Side Invariance means $I(M)=I(M^\top)$ for every $M$. Diagonal Monotonicity means $I(E^{ii})>I(E^{ij})$ for every $i\neq j$. Pairwise Marginal Monotonicity means that for every $i\neq j$ and every $\eps>0$, if $m_{ij}\geq\eps$ and $m_{ji}\geq\eps$, then
\[
I(M+\eps E^{ii}+\eps E^{jj}-\eps E^{ij}-\eps E^{ji})>I(M).
\]

\begin{proposition}[Implications of symmetry and monotonicity]\label{prop:linear-implications}
Suppose
\[
I(M)=\frac{\sum_{i,j}w_{ij}m_{ij}}{|M|}.
\]
Then:
\begin{enumerate}[label=(\alph*)]
\item Side Invariance holds if and only if $w_{ij}=w_{ji}$ for all $i,j$.
\item Diagonal Monotonicity holds if and only if $w_{ii}>w_{ij}$ for all $i\neq j$.
\item Pairwise Marginal Monotonicity holds if and only if $w_{ii}+w_{jj}>w_{ij}+w_{ji}$ for all $i\neq j$.
\end{enumerate}
\end{proposition}

\begin{proof}
For (a), if $I(M)=I(M^\top)$ for every $M$, applying this to $M=E^{ij}$ gives $w_{ij}=w_{ji}$. Conversely, if $w_{ij}=w_{ji}$ for all $i,j$, then
\[
\sum_{i,j}w_{ij}m_{ji}=\sum_{i,j}w_{ji}m_{ji}=\sum_{i,j}w_{ij}m_{ij},
\]
so $I(M^\top)=I(M)$. Part (b) follows directly from $I(E^{ij})=w_{ij}$. For (c), let
\[
\widetilde M:=M+\eps E^{ii}+\eps E^{jj}-\eps E^{ij}-\eps E^{ji}.
\]
Then $|\widetilde M|=|M|$, while the numerator changes by $\eps(w_{ii}+w_{jj}-w_{ij}-w_{ji})$. Thus $I(\widetilde M)>I(M)$ for all such $M$ and $\eps$ if and only if $w_{ii}+w_{jj}>w_{ij}+w_{ji}$ for every $i\neq j$.
\end{proof}

\subsection*{Unordered types}

This subsection specializes the linear-score representation to the case in which type labels have no intrinsic order.

When types are unordered, the labels $1,\ldots,n$ have no intrinsic meaning. Let $P$ be a permutation matrix. Type Permutation Invariance requires
\[
I(PMP^\top)=I(M)
\qquad(M\in\Nn).
\]

\begin{proposition}[Unordered normalized trace]\label{prop:unordered-trace}
Suppose $I$ has the linear-score representation and satisfies Type Permutation Invariance. Then there exist $a,b\in\R$ such that
\[
w_{ij}=\begin{cases}
a, & i=j,\\
b, & i\neq j.
\end{cases}
\]
If, in addition, Diagonal Monotonicity holds, then $a>b$. If the normalization $a=1$ and $b=0$ is imposed, then
\[
I(M)=\frac{\sum_{i=1}^n m_{ii}}{|M|}=\frac{\tr(M)}{|M|}.
\]
\end{proposition}

\begin{proof}
Type Permutation Invariance implies
\[
w_{ij}=I(E^{ij})=I(PE^{ij}P^\top)=I(E^{\pi(i)\pi(j)})
\]
for every permutation $\pi$. Any diagonal cell can be mapped to any other diagonal cell, so all diagonal weights are equal. Likewise, any off-diagonal cell can be mapped to any other off-diagonal cell, so all off-diagonal weights are equal. Diagonal Monotonicity gives $a>b$. The normalized formula follows by substituting $a=1$ and $b=0$.
\end{proof}

\subsection*{Ordered types}

This subsection specializes the linear-score representation to the case in which the distance between type labels is substantively meaningful.

When types are ordered, the labels have an ordinal meaning. Let $Q$ be the reversal permutation matrix, i.e., $Qe_i=e_{n+1-i}$. Type-Reversal Invariance requires
\[
I(QMQ^\top)=I(M).
\]
Distance Neutrality requires that if $|i-j|=|k-\ell|$, then $I(E^{ij})=I(E^{k\ell})$. Equivalently, the score of a cell depends only on the ordinal distance between the two matched types. Distance Monotonicity requires that if $|i-j|<|k-\ell|$, then $I(E^{ij})>I(E^{k\ell})$.

\begin{proposition}[Ordered normalized trace]\label{prop:ordered-trace}
Suppose $I$ has the linear-score representation and satisfies Distance Neutrality. Then there exists a function
\[
\phi:\{0,1,\ldots,n-1\}\to\R
\]
such that
\[
w_{ij}=\phi(|i-j|)
\]
and hence
\[
I(M)=\frac{\sum_{i,j}\phi(|i-j|)m_{ij}}{|M|}.
\]
If Distance Monotonicity holds, then $\phi$ is strictly decreasing. If the normalizations $\phi(0)=1$ and $\phi(n-1)=0$ are imposed, then perfect homogamy receives score one and maximal heterogamy receives score zero.
\end{proposition}

\begin{proof}
Define $\phi(d):=I(E^{ij})$ for any pair $(i,j)$ with $|i-j|=d$. This is well-defined by Distance Neutrality. Since $w_{ij}=I(E^{ij})$, we have $w_{ij}=\phi(|i-j|)$. The representation follows immediately. Distance Monotonicity says exactly that $\phi(d)>\phi(d')$ whenever $d<d'$, so $\phi$ is strictly decreasing.
\end{proof}

\begin{remark}
Type-Reversal Invariance alone implies $w_{ij}=w_{n+1-i,n+1-j}$. Distance Neutrality is stronger: it collapses all cells with the same distance $|i-j|$ into a common score.
\end{remark}

\section{Edge-probability bounds}\label{app:edge-bounds}

\begingroup
This appendix proves equation~\eqref{eq:many-edge-bounds}. Fix a man $u$ of degree $k$ and a woman $v$ of degree $\ell$. Partition $\mathcal G(d,e)$ into $\mathcal G_+$, containing the graphs with edge $uv$, and $\mathcal G_-$, containing those without it. Write $g_+:=|\mathcal G_+|$ and $g_-:=|\mathcal G_-|$, so $\pi_{k\ell}=g_+/(g_++g_-)$. All neighborhoods below refer to the graph before the switching.

For the upper bound, start from a graph in $\mathcal G_+$. Choose an edge $xy$ with $x\notin N(v)$ and $y\notin N(u)$, delete $uv,xy$, and add $uy,xv$. At most $D_U(\ell)+D_V(k)$ edges are excluded by these restrictions, so there are at least $L(1-\eta_{k\ell})_+$ valid choices. The resulting graph is in $\mathcal G_-$. In reverse, the edges $uy,xv$ can be chosen in at most $k\ell$ ways. Double counting gives
\[
g_+L(1-\eta_{k\ell})_+\le g_-k\ell,
\]
which rearranges to the upper bound in~\eqref{eq:many-edge-bounds}.

For the lower bound, start from a graph in $\mathcal G_-$. Choose edges $uy,xv$, and then choose an edge $zw$ such that $xw$ and $zy$ are absent. Delete $uy,xv,zw$ and add $uv,xw,zy$. These restrictions ensure that the three men and the three women are distinct, so the result is a simple graph in $\mathcal G_+$. For fixed $x,y$, at most $D_U(e_y)+D_V(d_x)$ choices of $zw$ are excluded. Since the interpolated degree-sum functions are increasing and concave,
\[
\frac1k\sum_{y\in N(u)}D_U(e_y)
\le D_U\!\left(\frac{D_V(k)}{k}\right),
\qquad
\frac1\ell\sum_{x\in N(v)}D_V(d_x)
\le D_V\!\left(\frac{D_U(\ell)}{\ell}\right).
\]
Thus each graph in $\mathcal G_-$ admits at least $k\ell L(1-\zeta_{k\ell})_+$ valid ordered switchings. In reverse, choose the two ordered edges $xw,zy$ in at most $L^2$ ways. Double counting gives
\[
g_-k\ell L(1-\zeta_{k\ell})_+\le g_+L^2,
\]
which rearranges to the lower bound. Both counting inequalities remain valid when one part of the partition is empty.
\endgroup

\section{Data details}\label{app:data-details}

\subsection{Organization of the platform}\label{app:platform-organization}

Our data come from IBJ, a platform that operates a common membership database shared by a large number of independently run marriage agencies. Figure \ref{fig:ibj_organization} illustrates this organization. Each user joins through a single agency, which checks the user's documents, helps prepare the profile, and arranges meetings; the agency, in turn, pays IBJ for access to the database. Because the database is shared, a user can search for, and be approached by, members of any affiliated agency rather than only the members of their own agency. Participation is costly: users pay registration and monthly fees, and age, marital status, education, and income are verified against official documents, and height is checked through health checkup records at some agencies and in person at meetings. The platform selects users who are actively seeking a spouse, with key attributes such as age and income checked against documents.

\begin{figure}[!htbp]
  \begin{center}
  \includegraphics[width = \textwidth]{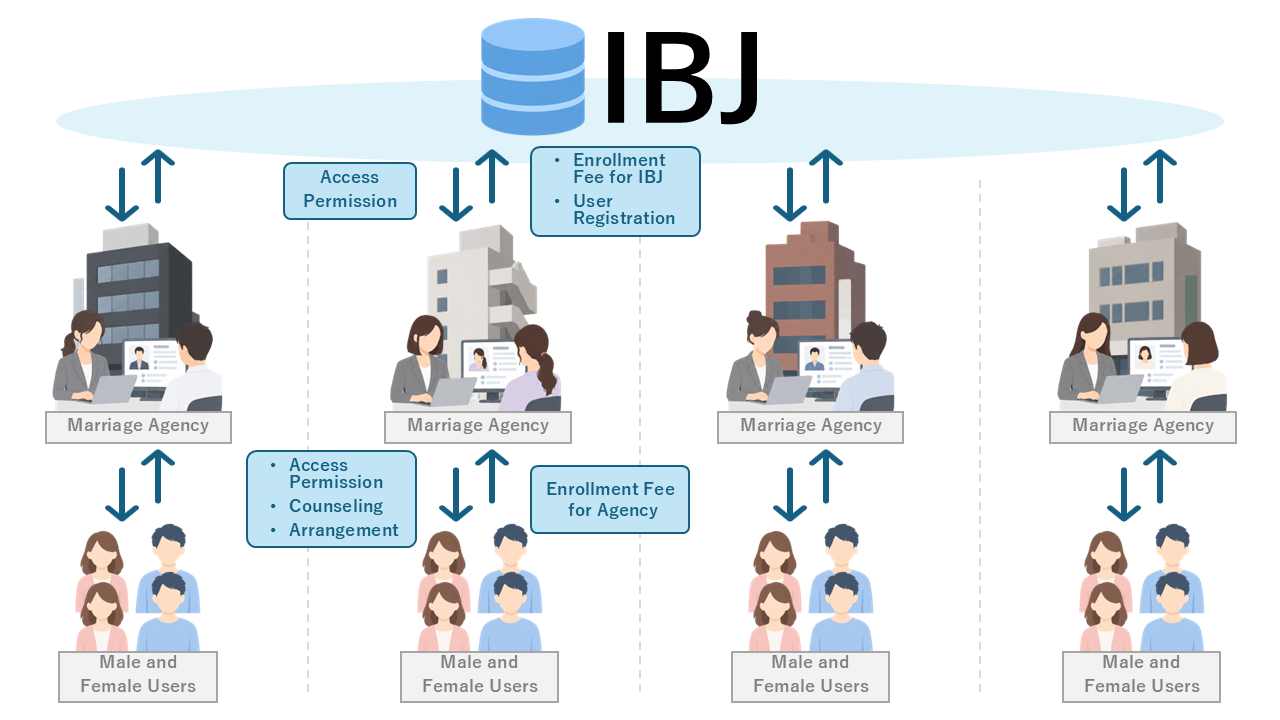}
  \caption{Organization of the IBJ Network}
  \label{fig:ibj_organization}
  \end{center}
\end{figure}

\subsection{Marginal distributions of users' attributes}\label{app:marginal-distributions}

Figure \ref{fg:marginal_distribution_by_year} shows that the marginal distributions of the three attributes are broadly stable over 2021--2024. The height distributions of men and women are almost unchanged across years. The male age distribution shifts toward younger ages, with larger shares in the early thirties and smaller shares in the mid-forties in 2024 than in 2021, while the female age distribution is stable. For income, the male distribution is stable, and the share of women in the lowest income category falls from about 32 percent in 2021 to about 26 percent in 2024. Because some marginals do move over time, we compare observed matching with a random benchmark computed from the observed marginals. ALR does not satisfy Marginal Independence, so these comparisons remain sensitive to changes in composition and categorization.

\begin{figure}[!htbp]
  \begin{center}
  \includegraphics[width = \textwidth]{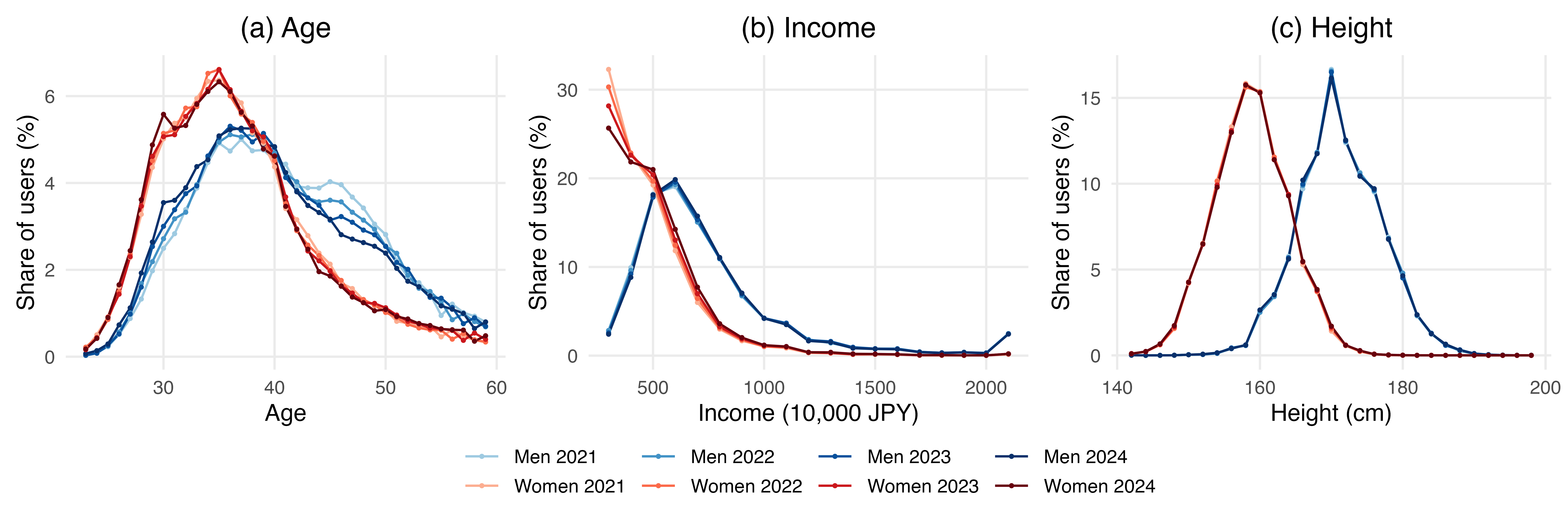}
  \caption{Marginal Distributions of Users' Attributes by Gender, 2021--2024}
  \label{fg:marginal_distribution_by_year}
  \end{center}
  \footnotesize
  Note: Each line shows the share of users in each bin among users of the given gender who sent or received at least one meeting request in the given year. Shares sum to 100 within each gender and year. Users whose attribute lies outside the bin range are excluded. Income is in units of 10,000 yen.
\end{figure}

\end{document}